\documentclass[preprint,11pt,authoryear]{elsarticle}
\usepackage{tikz}
\usepackage{xcolor}
\usepackage{amsfonts}
\usepackage{amsmath, calc}
\usepackage{amsthm}
\usepackage{amssymb}
\usepackage{times}
\usepackage{subfig}
\usepackage{multirow}
\usepackage{setspace}
\usepackage{enumerate}
\usepackage{natbib}
\usepackage{color}
\usepackage{url}
\usepackage{colortbl}
\usepackage{ wasysym }
\usepackage{verbatim}
\usetikzlibrary{shapes,arrows,calc,positioning}
 \usepackage{lscape}
\usepackage{array}
\usepackage{setspace}
\usepackage[normalem]{ulem}
\usepackage{arydshln}
\usepackage{extarrows}

\usepackage{caption}
\usepackage{rotating}

\newcolumntype{x}[1]{%
>{\centering\hspace{0pt}}p{#1}}%

\def\iid{\buildrel {\rm i.i.d.} \over \sim}

\def\i.i.d.{\buildrel {\rm i.i.d.} \over \sim}

\def\cw#1 { \overset{\mathbb{P}}{\underset{#1}{\longrightarrow}} }
\def\Real{\mathbb{R}}
\def\P#1{{\mathbb{P}}\left(#1\right)}

\def\Var#1{{\mathrm Var}\left(#1\right)}

\def \rcov#1#2 {{\rm cov}_{#1}\left( #2\right)}

\newcommand{\parallelsum}{\mathbin{\!/\mkern-5mu/\!}}

\newtheorem{example}{Example}

\newtheorem{lemma}{Lemma}
\newtheorem{theorem}{Theorem}
\newtheorem{definition}{Definition}
\newtheorem{corollary}{Corollary}
\newtheorem{remark}{Remark}

\newtheorem*{toy*}{Finite Model}

\newtheorem{model}{Model}
\newtheorem{model.c}{Hypothetical Model}

\newtheorem{single.c}{Single Year Hypothetical Model}

\begin{document}
\begin{frontmatter}
\title{On a Multi-Year Microlevel Collective Risk Model}

\author[EH2]{Rosy Oh\corref{cor1}}
\ead{rosy.oh5@gmail.com}
\author[UC]{Himchan Jeong\corref{cor1}}
\ead{himchan.jeong@uconn.edu}
\author[EH]{Jae Youn Ahn\corref{cor2}}
\ead{jaeyahn@ewha.ac.kr}
\author[UC]{Emiliano A. Valdez\corref{cor1}}
\ead{emiliano.valdez@uconn.edu}

\address[EH2]{Institute of Mathematical Sciences, Ewha Womans University, Seoul, Korea.}
\address[EH]{Department of Statistics, Ewha Womans University, Seoul, Korea.}
\address[UC]{Department of Mathematics, University of Connecticut, Connecticut, United States}

\cortext[cor1]{First Authors}
\cortext[cor2]{Corresponding Author}

\begin{abstract}


For a typical insurance portfolio, the claims process for a short period, typically one year, is characterized by observing frequency of claims together with the associated claims severities. The collective risk model describes this portfolio as a random sum of the aggregation of the claim amounts. In the classical framework, for simplicity, the claim frequency and claim severities are assumed to be mutually independent. However, there is a growing interest in relaxing this independence assumption which is more realistic and useful for the practical insurance ratemaking. While the common thread has been capturing the dependence between frequency and aggregate severity within a single period, the work of \cite{Oh2019copula} provides an interesting extension to the addition of capturing dependence among individual severities. In this paper, we extend these works within a framework where we have a portfolio of microlevel frequencies and severities for multiple years. This allows us to develop a factor copula model framework that captures various types of dependence between claim frequencies and claim severities over multiple years. It is therefore a clear extension of earlier works on one-year dependent frequency-severity models and on random effects model for capturing serial dependence of claims. We focus on the results using a family of elliptical copulas to model the dependence. The paper further describes how to calibrate the proposed model using illustrative claims data arising from a Singapore insurance company. The estimated results provide strong evidence of all forms of dependencies captured by our model.

\end{abstract}

\end{frontmatter}

\vfill

\pagebreak

\vfill

\pagebreak

\section{Introduction} \label{sec.1}

According to \cite{Klugman}, the aggregate loss in the classical collective risk model is defined as  $S = \sum_{i=1}^N Y_i,$ where $N$ means the number of claim and $Y_i$ denotes $i^{th}$ individual claim amounts over a fixed period of time with the following assumptions:
\begin{enumerate}
	\item  Conditional on $N = n$, the random variables $Y_{1}, Y_{2}, \ldots, Y_{n}$  are i.i.d. random
variables.
	\item  Conditional on $N = n$, the common distribution of the random variables $Y_{1}, Y_{2}, \ldots, Y_{n}$ does not depend on $n$.
	\item The distribution of $N$ does not depend in any way on the values of $Y_{i}$.
\end{enumerate}

These assumptions might be convenient in terms of computational ease, however, such simplifying assumptions often lead to bias issues especially when used for risk classification. In relaxing such assumptions, various models have been proposed in the insurance literature. An interesting method to model the dependence in the collective risk model is the so-called {\it two-part dependent frequency-severity model} as suggested by \citet{Frees2}. In this model, the dependence is incorporated by using frequency as an explanatory variable in the severity component. A similar approach has been used by \citet{Frees2011health} in the modeling and prediction of frequency and severity of health care expenditure. \citet{Peng} suggested a three-part framework in order to capture the association between frequency and severity components. When generalized linear models (GLMs) are used with the number of claims treated as a covariate in claims severity, \citet{Garrido} showed that the pure premium includes a correction term for inducing dependence. When analyzing bonus-malus data, an interesting observation was made by \citet{park2018} that dependence between claim frequency and severity is driven by the desire to reach a better bonus-malus class.

Applications of copula methods to capture dependence have been recently used in collective risk models. A majority of work in this area focused on modeling the dependence between frequency and average severity with parametric copulas. For example, \cite{czado2012mixed} used Gaussian copulas to extend traditional compound Poisson-Gamma two-part model and incorporated possible dependence. \cite{kramer2013} suggested a similar joint copula-based approach and interestingly observed that ignoring dependence causes a severe underestimation of total loss in a portfolio. \cite{Gee2016} extended the copula-based approach to dependent frequency and average severity using claims data with multiple lines of insurance business. While their findings suggested weak association between frequency and average severity, they concluded that there are strong dependencies among the lines of business.

Unlike choosing a suitable family of marginal distributions, it is usually much harder to choose the correct family of copulas when calibrating these dependent models with data. The work of \cite{kramer2013} investigated test procedures for the selection of a suitable family of copulas in a dependent frequency and average severity model. However, \cite{Oh2019copula} illustrated that indeed it is even more difficult to choose the appropriate dependence structure between frequency and average severity that includes the classical collective risk model as a special case. In particular, even under the most naive assumption of independence between frequency and individual severities, choosing the correct parametric copula presents some challenges.  Inspired by this phenomenon, \citet{Oh2019copula} and \citet{Cossette2019} discussed the construction of single year collective risk models with microlevel data to provide a suitable dependence structure between the frequency and severity components. In part, the extension in this paper that captures dependence of various types of dependence between claim frequency and claim severity over multiple years is motivated by the work of \cite{Oh2019copula}.

In insurance industry, it is important to model the longitudinal property of the insurance losses to predict the fair premium in the future based on each policyholder's historical claims information. However, the existing copula methods in the literature cannot be directly applied in prediction of the premium due to at least one of the following difficulties:
\begin{itemize}
  \item Limited to the analysis of data over a single period or cross-sectional data,
  \item The choice of the copula family to provide a suitable dependence structure between claim frequency and average claim severity can be difficult.
\end{itemize}
Alternatively, the random effect model can be used to model the longitudinal property of the insurance losses.
\citet{hernandez2009net} and \citet{PengAhn} used the shared random effects model to construct the dependence in a collective risk model, where independence between claim frequency and severity conditional on the random effect is assumed and the dependence structure is naturally derived by the shared random effects.
\citet{jeong2019predictive} derived a closed form of credibility premium for compound loss which captures not only the dependence between frequency and severity but also dependence among the multi-year claims of the same policyholder.
However, it is known that the overdispersion and serial dependence can be compounded in the random effect model. Such compounded effect of the random effect can possibly result in pseudo or fake dependence structure in the claims, which in turn leads to the poor prediction of the premium \citep{Denuit2, Murray2013, lee2020poisson}.

In this regard, as a natural extension of shared random effects model and one-year dependent compound risk model, we propose a multi-year framework with microlevel data so that we may incorporate the following dependencies simultaneously:
\begin{itemize}
  \item dependence between a frequency and a severity within a year,
  \item dependence between two distinct severities within a year,
  \item dependence among frequencies across years,
  \item dependence between a frequency and a severity in different years,
  \item dependence between two severities in different years.
\end{itemize}
Specifically, we use a factor copula representation, which can be viewed as a copula model version of the random effect model \citep{krupskii2013factorcopula, krupskii2015structured}, by using 1-year microlevel model as building blocks.

The remainder of this paper is organized as follows. In Section 2, we propose a generalized shared random effects framework for multi-year microlevel collective risk model that incorporates all types of dependencies previously described. We demonstrate that previous methods for dependence modeling can be considered as special cases of our proposed model. In Section 3, we provide a concrete example of our proposed model with elliptical copulas. Because of simplicity, we focus on the family of Gaussian copulas to further explore various correlation structures that satisfy our framework. In Section 4, an empirical analysis with a special case of our proposed model is conducted with a dataset from an automobile insurance company.  Concluding remarks are provided in Section 5 with some future directions of research.

\pagebreak

\section{Construction of the shared random effect parameter model} \label{sec.2}

\subsection{A motivating illustration}

While copula methods are flexible in modeling the dependence, the ``actual" flexibility comes from the proper choice of the parametric copula family.
Although one may consider using the nonparametric copula method for the full flexibility in choosing a copula structure, modeling and interpreting dependence based on the non-parametric copula can be difficult  as long as the discrete random variables are involved mainly due to the lack of uniquness \citep{genest2007primer}.
While recent study in \citet{yang2019nonparametric} provides the safe copula estimation method for discrete outcomes in a regression context, it is known to suffer from the so-called curse of dimensionality.

%


Indeed, as shown in \citet{Oh2019copula}, it is difficult to choose
a proper parametric copula family for the frequency and average severity even under the most naive assumption, the case where frequency and individual severities are independent. This subsection summarizes the example in \citet{Oh2019copula} to explain such difficult and the necessity to use microlevel claims information.

Consider the classical collective risk model where frequency $N$ and the individual severity $Y_j$s are assumed to be independent. Further, assume that
$N$ is a  positive integer valued random variable with
\[ \P{N=n}=\frac{1}{5}, \quad \hbox{for} \quad n=1,2,3,4,5, \]
and
\begin{equation}\label{eq.1}
Y_1, \cdots, Y_N \big\vert N \iid {\rm Gamma}(\xi, \psi).
\end{equation}
Then, \eqref{eq.1} implies
\[
M\big\vert N\sim {\rm Gamma}(\xi, \psi/N ).
\]
Clearly, $N$ and $M$ are not independent even though frequency and individual severities are independent.
Since $N$ is discrete, the visualization and interpretation of the corresponding copula density function for $(N, M)$ can be difficult. Alternatively, \citet{Oh2019copula} provides the density function for the jittered version of $(N, M)$ as shown in Figure \ref{figu.1} where x-axis and y-axis corresponds to frequency $N$ and the average severity $M$, respectively.


  \begin{figure}[h]
     \centering
    \subfloat{%
      \includegraphics[width=0.49\textwidth]{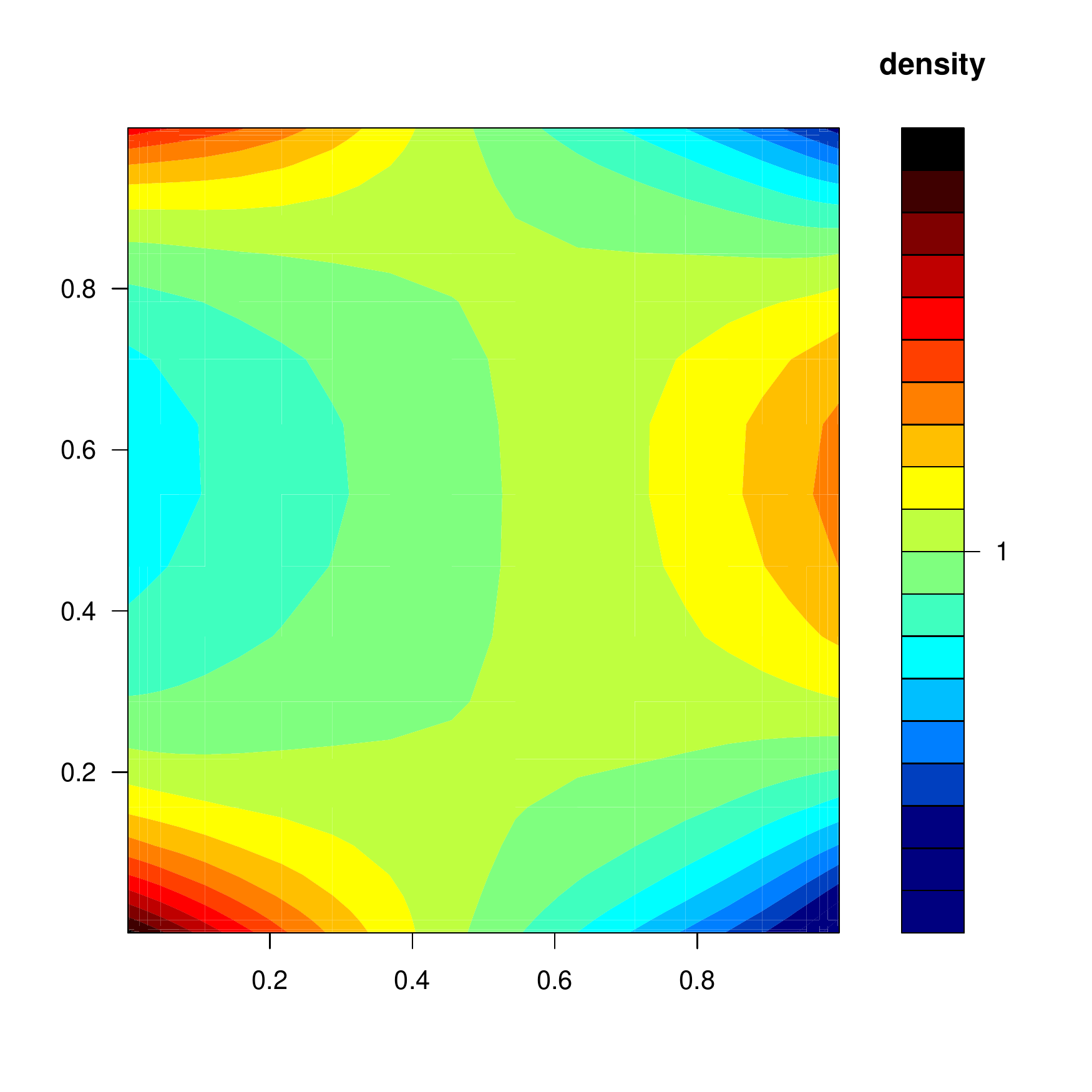}
    }
      \caption{Contour plot, in \citet{Oh2019copula}, of jittered copula density corresponding to $(N,M)$ using a kernel density estimation}
     \label{figu.1}
  \end{figure}

Let $(U_1, U_2)$ be a bivariate random vector sampled from the copula of the jittered version of $(N,M)$. As shown in Figure \ref{figu.1}, the density of the copula tends to be smaller in the middle part of $U_2$ when $U_1$ is smaller, wheras the density tends to be larger in the middle part of $U_2$ when $U_1$ is larger. Therefore, it is straightforward to see that conditional variance of $M$ decreases as $N$ increases in Figure \ref{figu.1}, which is quite intuitive since $\Var{M|N}= \xi^2 \psi /N$ in this case.

This example illustrates that we can see that most existing copulas, including Gaussian and Archimedean copulas, are unable to accommodate the dependence between frequency and average severity properly.
This is a motivation for the modeling the dependence based on the microlevel claims information rather than summarized claims information. We refer the readers to \citet{Oh2019copula} for more details of this example and the detailed construction of the jittered version of $(N, M)$.

\subsection{Data structure and model specification}

For non-life insurance, claims observed are typically a history of frequencies and severities for multiple years. For a policyholder observed for $\tau$ years, we have $n_1, \cdots, n_\tau$ which stand for frequency for each year, and corresponding individual severities $(\boldsymbol{y}_1, \cdots, \boldsymbol{y}_\tau)$ where
  \[
  \boldsymbol{y}_t=\begin{cases} \hbox{not defined}, & n_t=0;\\
  (y_{t,1}, \cdots, y_{t,{n_t}}), & n_t>0;
  \end{cases}
  \]
We find it convenient to define the following symbols for the description of data.

Define a random vector of length $N_t+1$
  \[
  \boldsymbol{Z}_t:=
  \begin{cases}
    \left( N_t, Y_{t,1}, \cdots, Y_{t, N_t}\right), & N_t\ge 1;\\
    0, & N_t=0,
  \end{cases}
  \]
  and the realization of $\boldsymbol{Z}_t$ is denoted as
  \[
  \boldsymbol{z}_t:=
  \begin{cases}
    \left( n_t, \boldsymbol{y}_{t} \right), & n_t\ge 1;\\
    0, & n_t=0.
  \end{cases}
  \]
 Furthermore,
 multi-year extension of $\boldsymbol{Z}_t$ is defined as
  \[
  \boldsymbol{Z}_{({\tau})}:=
    \left(\boldsymbol{Z}_1, \cdots, \boldsymbol{Z}_\tau  \right)
  \]
  and the realization of $\boldsymbol{Z}_{( \tau )}$ is denoted as
  \[
  \boldsymbol{z}_{( \tau )}:=
    \left(\boldsymbol{z}_1, \cdots, \boldsymbol{z}_\tau  \right).
  \]

In the subsequent, we describe a shared random effect parameter model for modeling the type of claims data we observe that primarily consist of frequencies and severities for multiple years.

\begin{model}[The copula linked shared random effect model]\label{model.1}
Consider the following random effect model for $\boldsymbol{Z}_t$ where the joint distribution between the observed losses and the shared random effect is presented with copulas.
\begin{enumerate}
  \item[i.] Shared random effect $R$ follows a probability distribution with density $\pi$.
  \item[ii.] Conditional on $R=r$, we have that
  $\boldsymbol{Z}_t \,$ for $t=1, \cdots$ are independent observations whose distribution function is given by
      \begin{equation}\label{eq.2}
      H_t\left( \boldsymbol{z}_t|r\right):=C_{(\theta_3, \theta_4)}\left( F_t(n_t|r),  G_{t,1}(y_{t,1}|r), \ldots, G_{t,n_t}(y_{t,n_t}|r) \right)
      \end{equation}
where  $F_t$ and $G_{t,j}$ means marginal cumulative distribution functions of $N_t$ and $Y_{t,j}$, respectively and $g_{t,j}$ means joint density function of $Y_{t,j}$.
      As a result, we have the following distribution function of $\boldsymbol{Z}_{(\tau)}$
      \[
      H(\boldsymbol{z}_{(\tau)}):=\int \prod\limits_{t=1}^{\tau} H_t\left( \boldsymbol{z}_t|r\right) \pi(r) {\rm d}r.
      \]
  \item[iii.] The parameters $\theta_3$ and $\theta_4$ of the copula $C_{(\theta_3, \theta_4)}$ controls the independence between the frequency and severities and independence among individual severities, respectively, within a year so that we have 
$$
{h_t{(\boldsymbol{z}_t|r)} } = f_t(n_t|r) g_t^{\rm [joint]}(\boldsymbol{y}_t|r) \ \text{ if and only if } \ \theta_3=0,
$$
where
$$
g_t{( \boldsymbol{y}_t|r)} =\prod\limits_{j=1}^{N_t} g_{t,j}{(y_{t,j}|r)}  \ \text{ if and only if } \ \theta_4=0,
$$
and
 $g_t^{\rm [joint]}$ means joint density function of $\boldsymbol{Y}_t$.

  \item[iv.] 
  $N_t \perp R \,$ for $t=1, \ldots$ if and only if $\theta_1=0$.
  \item[v.] 
  $\boldsymbol{y}_t \perp R\,$ for $t=1, \ldots$ if and only if $\theta_2=0$.
\end{enumerate}
\end{model}

Figure \ref{figu.2} illustrates the dependence structure of our proposed model. In this figure we show that shared random effect $R$ induces the types of dependence that are of interest to us. To illustrate, $R$ is linked to the number of claims across years, $(N_1, \ldots, N_\tau)$, through $C_{\theta_1}$ so that $\theta_1$ is a parameter which captures dependence among claim counts between years. Likewise, $R$ is linked to the individual amounts of claims across years, $(\boldsymbol{Y}_1, \ldots, \boldsymbol{Y}_\tau)$, through $C_{\theta_2}$ so that $\theta_2$ is a parameter which captures dependence among claim amounts within and across the years. Furthermore, $C_{\theta_1}$ combined with $C_{\theta_2}$ introduces the dependence between the claim counts and individual severities within and across the years.

While, via the shared random effect $R$, the parameters $\theta_1$ and $\theta_2$ universally capture dependence among the claims across the years, the other parameters $\theta_3$ and $\theta_4$ specifically capture dependence within the claims of the same year. That is, $\theta_3$ is a parameter which incorporates the dependence between the claim count and claim amounts within a year whereas $\theta_4$ is a parameter which incorporates the dependence among claim amounts within a year. Similarly, $\theta_3$ combined with $\theta_4$ affects the dependence between the claim counts and individual severities within the year.
As a result, while dependence among the claims in different years are modeled by $(\theta_1, \theta_2)$ only, the dependence among the claims in the same year are modeled by both $(\theta_1, \theta_2)$ and $(\theta_3, \theta_4)$. Note that our framework is distinguished from some existing work on dependence modeling with copulas such as   \citet{shi2018pair} and \cite{Leegee}, where average severity in the form of summarized data was used for modeling and implicitly precluded independence among the individual severities within the same year.

$\,$

The idea of our multi-year microlevel collective risk model is that the observed claim for year $t$, $Z_{t}$, are independent for $t=1, \ldots, \tau$ given the shared random effect $R=r$ described as follows:
  \[
h{(\boldsymbol{z}_{(\tau)} |r)}  =\prod\limits_{t=1}^{\tau} h_t({z_t|r}) \xLongrightarrow{\theta_3=0} \prod\limits_{t=1}^{\tau}
\left[  f_t(n_t |r) g_t^{\rm [joint]} \left( \boldsymbol{y}_{t} |r\right)      \right] \xLongrightarrow{\theta_3=\theta_4=0} \prod\limits_{t=1}^{\tau}
\left[  f_t(n_t |r)  \left \{ \prod_{j=1}^{n_t} g_{t,j} \left( {y}_{t,j} |r\right)  \right \}    \right],
  \] and
\begin{equation} \label{eq.3}
\begin{aligned}
h(\boldsymbol{z}_{(\tau)})&=\int  h{(\boldsymbol{z}_{(\tau)} |r)} \pi(r){\rm d}{r}  \xLongrightarrow{\theta_3=0}  \int  \prod\limits_{t=1}^{\tau}
\left[  f_t(n_t |r) g_t^{\rm [joint]} \left( \boldsymbol{y}_{t} |r\right)      \right] \pi(r){\rm d}{r}\\
&\xLongrightarrow{\theta_3=\theta_4=0} \int \prod\limits_{t=1}^{\tau}
\left[  f_t(n_t |r)  \left \{ \prod_{j=1}^{n_t} g_{t,j} \left( {y}_{t,j} |r\right)  \right \}    \right] \pi(r){\rm d}{r} ,
\end{aligned}
\end{equation}
which is straightforward from iii and iv of Model \ref{model.1}.

  \begin{figure}[p]
  \begin{center}
      \includegraphics[width=1\textwidth]{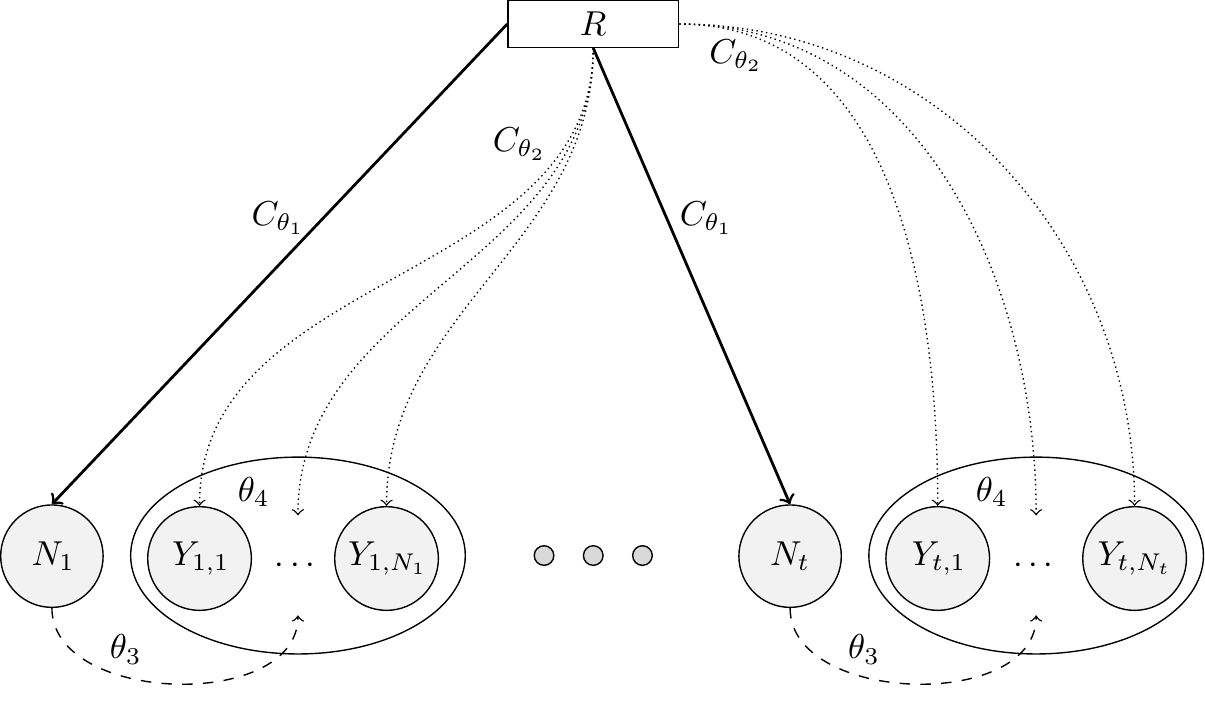}
       \caption{Visual representation of the multi-year microlevel shared random effect model}
     \label{figu.2}
 \end{center}
  \end{figure}

We note that this construction is similar to the model described by \cite{krupskii2013factorcopula}, which develops a factor copula model conditionally on a set of latent variables. In some sense, according to their paper, our approach leads to a one-factor copula model presented in Section \ref{sec.3}. The primary difference in our approach is the clear intuitive interpretation of our model to describe the various types of dependence in a dependent collective risk model. The well-definedness of Model \ref{model.1} will also be discussed in Remark \ref{rem.1} in Section \ref{sec.3}.

\subsection{Special cases}

It is immediate to see that the classical collective risk model of \cite{Klugman} is a special case of our proposed model where $\theta_1=\theta_2=\theta_3=\theta_4=0$. This is the case when all frequencies and severities are mutually independent. \cite{baumgartner2015sharedre} proposed shared random effects model to capture association between frequency and the average severity, which is just another special case of our proposed model. This is the case when $\theta_3=\theta_4=0$.
Finally, it is also easy to check that single-year microlevel collective risk model, proposed by \citet{Oh2019copula}, is another special case of our proposed model. This is when $\theta_1=\theta_2=0$.
In this regard, our proposed framework is quite comprehensive that allows other dependence models that have appeared in the literature as special cases.


\section{Factor copula model based on the elliptical distributions} \label{sec.3}

Copulas generated by elliptical distributions, also called \textit{elliptical copulas}, have the correlation matrix as the primary parameter describing dependence between the components. The Gaussian and $t$ copulas belong to the family of elliptical copulas. We refer to \citet{landsman2003tail} for other choices of elliptical copulas including the copulas generated from multivariate Cauchy or multivariate logistic distribution.
 In this section, for simplicity, apparent ease of computations, and steering clear of distractions from the general case, we focus on the case of Gaussian copulas. In  \ref{appendix.b}, we illustrate how Gaussian copulas in multi-year microlevel collective risk model can be generalized into the elliptical copulas by providing an example of $t$ copula among other choices of elliptical copulas.  Specifically, we consider Gaussian copulas with a specific covariance matrix to accommodate the dependence structure of multi-year microlevel collective risk model, and show that such Gaussian copula models can be represented as factor copula models.
For the use in elliptical copulas including the Gaussian and $t$ copulas in mind, we begin with describing dependence structure via correlation matrices.

\subsection{Dependence structure via correlation matrix}

We start with definition of symbols. Denote $\mathbb{N}$, $\mathbb{N}_0$, $\Real$, and $\Real^+$ by the set of positive integer, the set of non-negative integer, the set of real number, and the set of positive real number, respectively.

For a $n\times m$ matrix $\boldsymbol{M}$, we denote $(i,j)$-th component of $\boldsymbol{M}$ as $\left[ \boldsymbol{M}\right]_{ij}$. For a row vector $\boldsymbol{v}$ of length $n$, we denote the $i$-th component of $\boldsymbol{v}$ as $\left[\boldsymbol{v}\right]_i$.
    For $n\in\mathbb{N}$, define ${\bf 1}_n$ and $\boldsymbol{J}_{n\times n}$ as a column vector of $1$ with length $n$ and a $n \times n$ matrix of ones, respectively. We use $\boldsymbol{I}_n$ for $n\in\mathbb{N}$ to represent the $n \times n$ identity matrix.

Suppose $\boldsymbol{\Sigma}_{1,1}$, $\boldsymbol{\Sigma}_{1,2}$, $\boldsymbol{\Sigma}_{2,1}$, and $\boldsymbol{\Sigma}_{2,2}$ are $\ell\times \ell$, $\ell\times m$, $m\times \ell$, and $m\times m$ matrices, respectively.
Define $(\ell+m)\times (\ell+m)$ matrix $\boldsymbol{\Sigma}$ as
\[
\boldsymbol{\Sigma}=\left(
                      \begin{array}{cc}
                        \boldsymbol{\Sigma}_{1,1} & \boldsymbol{\Sigma}_{1,2} \\
                        \boldsymbol{\Sigma}_{2,1} & \boldsymbol{\Sigma}_{2,2} \\
                      \end{array}
                    \right)
\]
If $\boldsymbol{\Sigma}_{2,2}$ is invertible, the Schur complement of the block $\boldsymbol{\Sigma}_{2,2}$
of the matrix $\boldsymbol{\Sigma}$ is the $\ell\times \ell$ matrix defined by
\[
\boldsymbol{\Sigma}\parallelsum \boldsymbol{\Sigma}_{1,1} :=\boldsymbol{\Sigma}_{2,2}- \boldsymbol{\Sigma}_{2,1}\left( \boldsymbol{\Sigma}_{1,1} \right)^{-1} \boldsymbol{\Sigma}_{1,2}.
\]

\begin{definition}\label{def.1}
For $\boldsymbol{n}= (n_1, \cdots, n_\tau)\in\left(\mathbb{N}_0\right)^\tau$ and $\boldsymbol{\rho}=(\rho_1, \cdots, \rho_5)\in \Real^5$
define the partitioned matrix
\begin{equation}\label{eq.4}
\boldsymbol{\Sigma}_{(\boldsymbol{n})}^{(\boldsymbol{\rho})}:=\left(                     \begin{array}{ccc}
                        \boldsymbol{\Sigma}_{11}^{(\boldsymbol{\rho})} & \cdots & \boldsymbol{\Sigma}_{1\tau}^{(\boldsymbol{\rho})}\\
                               \vdots & \ddots & \vdots\\
                        \boldsymbol{\Sigma}_{\tau 1}^{(\boldsymbol{\rho})} & \cdots & \boldsymbol{\Sigma}_{\tau\tau}^{(\boldsymbol{\rho})}\\
                      \end{array} \right).
\end{equation}
For $i=1, \cdots, \tau$, the matrix $\boldsymbol{\Sigma}_{tt}^{(\boldsymbol{\rho})}$ is a $(n_t+1)\times (n_t+1)$ matrix defined as
\[
\left[\boldsymbol{\Sigma}_{tt}^{(\boldsymbol{\rho})}\right]_{\ell m}=\begin{cases}
  1, &\ell=m;\\
  \rho_2, & \ell\neq m,\quad \min\{\ell, m\}\ge 2;\\
  \rho_1, & \hbox{elsewhere};
\end{cases}
\]
for $\ell,m=1, \cdots, n_t+1$.
Furthermore, for $i,j=1, \cdots, t$ with $i\neq j$,  the matrix $\boldsymbol{\Sigma}_{ij}$ is a $(n_i+1)\times (n_j+1)$ matrix defined as
\[
\left[\boldsymbol{\Sigma}_{tj}^{(\boldsymbol{\rho})}\right]_{\ell m}=\begin{cases}
  \rho_3, &\ell=m=1;\\
  \rho_5, & \min\{\ell, m\}\ge 2;\\
  \rho_4, & \hbox{elsewhere};
\end{cases}
\]
for $\ell=1, \cdots, n_t+1$ and $m=1, \cdots, n_j+1$.
\end{definition}
\begin{example}
Consider the case $\boldsymbol{n}= (2, 3)$. Then we can write out $\boldsymbol{\Sigma}_{(\boldsymbol{n})}^{(\boldsymbol{\rho})}$ by denoting  $\boldsymbol{n}= (2, 3)$ and $\boldsymbol{\rho}=(\rho_1, \cdots, \rho_5)\in \Real^5$. As a result,
  $\boldsymbol{\Sigma}_{(\boldsymbol{n})}^{(\boldsymbol{\rho})}$ is a $7\times 7$ defined as
$$
\boldsymbol{\Sigma}_{(\boldsymbol{n})}^{(\boldsymbol{\rho})}:=\left(                     \begin{array}{cc}
                        \boldsymbol{\Sigma}_{11}^{(\boldsymbol{\rho})} & \boldsymbol{\Sigma}_{12}^{(\boldsymbol{\rho})}\\
                        \boldsymbol{\Sigma}_{21}^{(\boldsymbol{\rho})} &  \boldsymbol{\Sigma}_{22}^{(\boldsymbol{\rho})}\\
                      \end{array} \right)
$$
where
\[
\boldsymbol{\Sigma}_{11}^{(\boldsymbol{\rho})}=\left(
                                                 \begin{array}{ccc}
                                                   1 & \rho_1 & \rho_1 \\
                                                   \rho_1 & 1 & \rho_2 \\
                                                   \rho_1 & \rho_2 & 1 \\
                                                 \end{array}
                                               \right)
                                               \quad
\boldsymbol{\Sigma}_{12}^{(\boldsymbol{\rho})}=\left(
                                                 \begin{array}{cccc}
                                                  \rho_3 & \rho_4 & \rho_4 & \rho_4 \\
                                                   \rho_4 & \rho_5 & \rho_5 & \rho_5\\
                                                   \rho_4 & \rho_5 & \rho_5 & \rho_5\\
                                                 \end{array}
                                               \right)
                                               \quad\hbox{and}\quad
\boldsymbol{\Sigma}_{22}^{(\boldsymbol{\rho})}=\left(
                                                 \begin{array}{cccc}
                                                   1 & \rho_1 & \rho_1 & \rho_1 \\
                                                   \rho_1 & 1 & \rho_2 & \rho_2\\
                                                   \rho_1  & \rho_2 & 1 & \rho_2\\
                                                   \rho_1 & \rho_2 & \rho_2 & 1 \\
                                                 \end{array}.
                                               \right)
\]
Furthermore, from the above and the following
\[
\boldsymbol{\Sigma}_{21}^{(\boldsymbol{\rho})}=\left( \boldsymbol{\Sigma}_{12}^{(\boldsymbol{\rho})} \right)^{\mathrm T},
\]
 we have
 \[
 \boldsymbol{\Sigma}_{(\boldsymbol{n})}^{(\boldsymbol{\rho})}=
 \left(
   \begin{array}{ccc:cccc}
     1 & \rho_1 & \rho_1 & \rho_3 & \rho_4 & \rho_4 & \rho_4 \\
     \rho_1 & 1 & \rho_2 & \rho_4 & \rho_5 & \rho_5 & \rho_5 \\
     \rho_1 & \rho_2 & 1 & \rho_4 & \rho_5 & \rho_5 & \rho_5 \\
         \hdashline
     \rho_3 & \rho_4 & \rho_4 & 1 & \rho_1 & \rho_1 & \rho_1 \\
     \rho_4 & \rho_5 & \rho_5 & \rho_1 & 1 & \rho_2 & \rho_2  \\
     \rho_4 & \rho_5 & \rho_5 & \rho_1 & \rho_2 & 1 & \rho_2  \\
     \rho_4 & \rho_5 & \rho_5 & \rho_1 & \rho_2 & \rho_1 & 1 \\
   \end{array}
 \right)
 \]

\end{example}
In the matrix $\boldsymbol{\Sigma}_{(\boldsymbol{n})}^{(\boldsymbol{\rho})}$, each component will be used for modeling the correlation between frequencies and severities within and across years. For example, the partitioned matrix $\boldsymbol{\Sigma}_{tt}^{(\boldsymbol{\rho})}$ is a $(n_t+1)\times (n_t+1)$ matrix describing the correlation structure of the random vector $(N_t, Y_{t,1}, \cdots, Y_{t, n_t})$. Specifically, $\rho_1$ in $\boldsymbol{\Sigma}_{tt}^{(\boldsymbol{\rho})}$ is used for a correlation between a frequency and a severity in the $t$-th year, and $\rho_2$ in
$\boldsymbol{\Sigma}_{tt}^{(\boldsymbol{\rho})}$ is used for a correlation among the severities in the $t$-th year.
 On the other hand, the partitioned matrix $\boldsymbol{\Sigma}_{tj}^{(\boldsymbol{\rho})}$ is a $(n_t+1)\times (n_j+1)$ matrix describing the correlation structure between the random vectors $(N_t, Y_{t,1}, \cdots, Y_{t, n_t})$ and $(N_j, Y_{j,1}, \cdots, Y_{j, n_j})$. Specifically,
$\rho_3$ in $\left[\boldsymbol{\Sigma}_{tj}^{(\boldsymbol{\rho})}\right]_{1 1}$ is used for a correlation between the frequencies in the different years, and $\rho_4$ in $\boldsymbol{\Sigma}_{tj}^{(\boldsymbol{\rho})}$ is used for a correlation between a frequency in different years. Finally, $\rho_5$ in $\boldsymbol{\Sigma}_{tj}^{(\boldsymbol{\rho})}$ is used for a correlation between a frequency and a severity in different years. The following is  summarization for the meaning of each correlation:
\begin{itemize}
  \item $\rho_1$: correlation between a frequency and a severity within a year;
  \item $\rho_2$: correlation among two distinct severities within a year;
  \item $\rho_3$: correlation among frequencies across years;
  \item $\rho_4$: correlation between a frequency and a severity in different years;
  \item $\rho_5$: correlation between two severities in different years.
\end{itemize}

We finally note that $\boldsymbol{\Sigma}_{tt}^{(\boldsymbol{\rho})}$ only depends on $(\rho_1, \rho_2)$ while $\boldsymbol{\Sigma}_{tj}^{(\boldsymbol{\rho})}$ for $t\neq j$ only depends on $(\rho_3, \rho_4, \rho_5)$. Hence, we find that it is convenient to use $\boldsymbol{\Sigma}_{tt}^{(\boldsymbol{\rho}^*)}$ with
$\boldsymbol{\rho}^*=(\rho_1, \rho_2)$
to stand for $\boldsymbol{\Sigma}_{tt}^{(\boldsymbol{\rho})}$, and similarly
$\boldsymbol{\Sigma}_{tj}^{(\boldsymbol{\rho}^*)}$ for $t\neq j$ with $\boldsymbol{\rho}^*=(\rho_3, \rho_4, \rho_5)$ to stand for $\boldsymbol{\Sigma}_{tj}^{(\boldsymbol{\rho})}$ in a clear context.

\begin{definition}\label{def.2}
  For $\boldsymbol{n}= (n_1, \cdots, n_\tau)\in\mathbb{N}_0^\tau$, $\boldsymbol{\rho}=(\rho_1, \cdots, \rho_5)\in (-1,1)^5$,
  $\boldsymbol{\theta}=(\theta_1, \theta_2)\in (-1,1)^2$,
define the partitioned matrix as $\boldsymbol{\Sigma}_{(\boldsymbol{n})}^{(\boldsymbol{\rho}; \boldsymbol{\theta})}$
as
\begin{equation}\label{eq.5}
\boldsymbol{\Sigma}_{(\boldsymbol{n})}^{(\boldsymbol{\rho}, \boldsymbol{\theta})}:=
\left(
  \begin{array}{cc}
    \boldsymbol{I}_1 & \boldsymbol{\Omega}_{(\boldsymbol{n})}^{(\boldsymbol{\theta})} \\
    \left( \boldsymbol{\Omega}_{(\boldsymbol{n})}^{(\boldsymbol{\theta})}\right)^{\mathrm T} & \boldsymbol{\Sigma}_{(\boldsymbol{n})}^{(\boldsymbol{\rho})}\\
  \end{array}
\right)
\end{equation}
where $\boldsymbol{\Sigma}_{(\boldsymbol{n})}^{(\boldsymbol{\rho})}$ is defined in \eqref{eq.4} and $\boldsymbol{\Omega}_{(\boldsymbol{n})}^{(\boldsymbol{\theta})}$ is a $1\times (\bar{\boldsymbol{n}}+\tau)$  matrix which can be expressed based on the following partitioned matrix
\[
\boldsymbol{\Omega}_{(\boldsymbol{n})}^{(\boldsymbol{\theta})}:= \left(\boldsymbol{\Omega}_{n_1}^{(\boldsymbol{\theta})}, \cdots, \boldsymbol{\Omega}_{n_\tau}^{(\boldsymbol{\theta})} \right)
\]
with $\boldsymbol{\Omega}_{n_t}^{(\boldsymbol{\theta})}$ being a $1\times (n_t+1)$ matrix given by
\[
\left[ \boldsymbol{\Omega}_{n_t}^{(\boldsymbol{\theta})} \right]_{1\ell}:=
\begin{cases}
  \theta_1, & \ell=1;\\
  \theta_2, & \hbox{otherwise}.
\end{cases}
\]
\end{definition}

In Definition \ref{def.2}, we have introduced two parameters $\theta_1$ and $\theta_2$. We impose natural dependence for multiples years of observed claims by using the shared random effect  $R$, which will affect all frequency and severities in any calendar year. In this regard, $\theta_1$ will be served as correlation parameter between the random effect $R$ and a frequency, and $\theta_2$ will be served as correlation parameter between a random effect $R$ and each severity, as described in Figure \ref{figu.1}.

\begin{example}
Consider the case $\boldsymbol{n}= (2, 3)\in\mathbb{N}_0^2$, then one can represent
$\boldsymbol{\Sigma}_{\boldsymbol{n}}^{(\boldsymbol{\rho}, \boldsymbol{\theta})}$ as a partitioned matrix as
  \[
\boldsymbol{\Sigma}_{(\boldsymbol{n})}^{(\boldsymbol{\rho}, \boldsymbol{\theta})}:=
\left(
  \begin{array}{cc}
    \boldsymbol{I}_1 & \boldsymbol{\Omega}_{(\boldsymbol{n})}^{(\boldsymbol{\theta})} \\
    \left( \boldsymbol{\Omega}_{(\boldsymbol{n})}^{(\boldsymbol{\theta})}\right)^{\mathrm T} & \boldsymbol{\Sigma}_{(\boldsymbol{n})}^{(\boldsymbol{\rho})}\\
  \end{array}
\right)
\]
where $\boldsymbol{\Sigma}_{\boldsymbol{n}}^{(\boldsymbol{\rho})}$ is in \eqref{eq.4}, and
\[
\boldsymbol{\Omega}_{\boldsymbol{n}}^{(\boldsymbol{\theta})}=
\left(
  \begin{array}{ccccccc}
    \theta_1 & \theta_2 & \theta_2 & \theta_1 & \theta_2 & \theta_2 & \theta_2 \\
  \end{array}
\right).
\]
Hence, we have
\[
\boldsymbol{\Sigma}_{(\boldsymbol{n})}^{(\boldsymbol{\rho}, \boldsymbol{\theta})}=
\left(
  \begin{array}{c:ccc:cccc}
    1 & \theta_1 & \theta_2 & \theta_2 & \theta_1 & \theta_2 & \theta_2 & \theta_2 \\
             \hdashline
    \theta_1 &  1 & \rho_1 & \rho_1 & \rho_3 & \rho_4 & \rho_4 & \rho_4 \\
    \theta_2 &  \rho_1 & 1 & \rho_2 & \rho_4 & \rho_5 & \rho_5 & \rho_5 \\
    \theta_2 &  \rho_1 & \rho_2 & 1 & \rho_4 & \rho_5 & \rho_5 & \rho_5 \\
         \hdashline
    \theta_1 & \rho_3 & \rho_4 & \rho_4 & 1 & \rho_1 & \rho_1 & \rho_1 \\
    \theta_2 & \rho_4 & \rho_5 & \rho_5 & \rho_1 & 1 & \rho_2 & \rho_2  \\
    \theta_2 & \rho_4 & \rho_5 & \rho_5 & \rho_1 & \rho_2 & 1 & \rho_2  \\
    \theta_2 & \rho_4 & \rho_5 & \rho_5 & \rho_1 & \rho_2 & \rho_1 & 1 \\
  \end{array}
\right).
\]
\end{example}

Now, for $\boldsymbol{n}= (n_1, \cdots, n_\tau)\in\left(\mathbb{N}_0\right)^\tau$ and $\boldsymbol{\rho}=(\rho_1, \cdots, \rho_5)\in \Real^5$, we consider reparameterization of a matrix $\boldsymbol{\Sigma}_{(\boldsymbol{n})}^{(\boldsymbol{\rho})}$ with
\begin{equation}\label{eq.6}
\begin{cases}
  \rho_1=\theta_1\theta_2+\theta_3\theta_4\\
  \rho_2=\theta_2^2+\theta_4^2\\
  \rho_3=\theta_1^2\\
  \rho_4=\theta_1\theta_2\\
  \rho_5=\theta_2^2
\end{cases}
\end{equation}
for
\[
\boldsymbol{\theta}=\left( \theta_1, \cdots, \theta_4\right)\in\Real^4.
\]
The following theorem provides some results related with reparameterization in \eqref{eq.6}.

\begin{theorem} \label{thm.1}
    For $\boldsymbol{n}= (n_1, \cdots, n_\tau)\in\mathbb{N}_0^\tau$, $\boldsymbol{\rho}=(\rho_1, \cdots, \rho_5)\in (-1,1)^5$, consider the Schur Complement of the block $\boldsymbol{I}_1$ of the matrix $\boldsymbol{\Sigma}_{(\boldsymbol{n})}^{(\boldsymbol{\rho}, \boldsymbol{\theta})}$ in \eqref{eq.5}
    denoted as $\boldsymbol{M}:=\boldsymbol{\Sigma}_{(\boldsymbol{n})}^{(\boldsymbol{\rho}, \boldsymbol{\theta})}\parallelsum\boldsymbol{I}_1$. For convenience, consider the following block matrix representation of $\boldsymbol{M}$ as
    \begin{equation}\label{eq.7}
     \boldsymbol{M}=
    \left(                     \begin{array}{ccc}
                        \boldsymbol{M}_{1 1} & \cdots & \boldsymbol{M}_{1\tau}\\
                               \vdots & \ddots & \vdots\\
                        \boldsymbol{M}_{\tau 1} & \cdots & \boldsymbol{M}_{\tau\tau}\\
                      \end{array} \right)
    \end{equation}
    where $\boldsymbol{M}_{ij}$ is a $n_i \times n_j$ matrix. Then, we have the following results.
    \begin{enumerate}
      \item[i.] For any $\boldsymbol{n}\in\left(\mathbb{N}_0\right)^{\tau}$, $\boldsymbol{M}$ is a block diagonal matrix, i.e. $\boldsymbol{M}_{ij}$ is a $n_i \times n_j$ zero matrix whenever $i \neq j$, if and only if $\rho_3$, $\rho_4$, and $\rho_5$ satisfy
          \begin{equation}\label{eq.8}
            \rho_3=\theta_1^2, \quad \rho_4=\theta_1\theta_2, \quad\hbox{and}\quad \rho_5=\theta_2^2.
          \end{equation}
      \item[ii.] A matrix $\boldsymbol{\Sigma}_{(\boldsymbol{n})}^{(\boldsymbol{\rho}, \boldsymbol{\theta})}$ is positive definite and $\boldsymbol{M}$ is a block diagonal matrix
          for any $\boldsymbol{n}\in\left(\mathbb{N}_0\right)^{\tau}$
          if and only if $\boldsymbol{\rho}$ is represented as in \eqref{eq.6} and satisfying
          \begin{equation}\label{eq.9}
          \theta_1^2+\theta_3^2<1 \quad\hbox{and}\quad \theta_2^2+\theta_4^2<1.
          \end{equation}
      \item[iii.] 
      A matrix $\boldsymbol{\Sigma}_{(\boldsymbol{n})}^{(\boldsymbol{\rho})}$ with the parametrization in \eqref{eq.6} is positive definite for any $\boldsymbol{n}\in\left(\mathbb{N}_0\right)^{\tau}$ if
      $\boldsymbol{\theta}$ satisfies \eqref{eq.9}.
    \end{enumerate}
\end{theorem}

\begin{proof}
For the proof of part i, it suffices to show that if $i \ne j$, then
$$
\begin{aligned}
\boldsymbol{M}_{ij} =  \boldsymbol{\Sigma}_{ij} - [\boldsymbol{\Omega}_{n_i}^{(\boldsymbol{\theta})}]^{\mathrm T} \boldsymbol{\Omega}_{n_j}^{(\boldsymbol{\theta})}
\end{aligned}
$$
by definition of $\boldsymbol{M}$ where $\boldsymbol{\Sigma}_{ij}$ and $\boldsymbol{\Omega}_{n_t}^{(\boldsymbol{\theta})}$ are defined in \eqref{eq.5} and \eqref{eq.7}, respectively and it can be written as follows:
\[
\left[\boldsymbol{M}_{ij} \right]_{\ell m}=\begin{cases}
  \rho_3 - \theta_1^2, &\ell=m=1;\\
  \rho_5 - \theta_2^2, & \min\{\ell, m\}\ge 2;\\
  \rho_4 - \theta_1\theta_2, & \hbox{elsewhere}.
\end{cases}
\]
  For the proof of part ii, by Schur decomposition, we have $\boldsymbol{\Sigma}_{(\boldsymbol{n})}^{(\boldsymbol{\rho}, \boldsymbol{\theta})}$ is positive definite if and only if $\boldsymbol{M}$ is positive definite.
  Since $\boldsymbol{M}$ is a block diagonal matrix provided \eqref{eq.6} is satisfied, checking the positive definiteness of $\boldsymbol{M}$ is equivalent to check whether $\boldsymbol{M}_{jj}$
  is positive definite or not. 
  Hence, a matrix $\boldsymbol{\Sigma}_{(\boldsymbol{n})}^{(\boldsymbol{\rho}, \boldsymbol{\theta})}$ is positive definite and $\boldsymbol{M}$ is a block diagonal matrix
          for any $\boldsymbol{n}\in\left(\mathbb{N}_0\right)^{\tau}$
           if and only if
           $\boldsymbol{\Sigma}_{(tt)}^{(\boldsymbol{\rho}^*)}$ is positive definite for any  $t\in\mathbb{N}_0$ where
  \begin{equation*}
  \rho_1^*=\frac{\rho_1-\theta_1\theta_2}{\sqrt{1-\theta_1^2}\sqrt{1-\theta_2^2}}
  \quad\hbox{and}\quad
  \rho_2^*=\frac{\rho_2-\theta_2^2}{1-\theta_2^2}, \quad \theta_1, \theta_2\in(-1,1).
  \end{equation*}
 Following Corollary 1 in \citet{Oh2019copula}, we have positive definite $\boldsymbol{\Sigma}_{(tt)}^{(\boldsymbol{\rho}^*)}$
  for any $t\in\mathbb{N}_0$ if and only if
  \begin{equation}\label{eq.10}
  \left(\rho_1^*\right)^2<\rho_2^*<1.
  \end{equation}
  Finally, simple argument shows that \eqref{eq.10} with the condition $\theta_1, \theta_2\in(-1,1)$ is equivalent with
  \[
  \rho_1=\theta_1\theta_2+\theta_3\theta_4\quad\hbox{and}\quad
  \rho_2=\theta_2^2+\theta_4^2
  \]
  for
            \[
          \theta_1^2+\theta_3^2<1 \quad\hbox{and}\quad \theta_2^2+\theta_4^2<1.
          \]
The proof of part iii immediately follows from part i and ii.
\end{proof}

\subsection{The special case of Gaussian copulas}\label{sec.3.2}

Let 
$F_t$ be non-negative integer-valued distribution functions with the respective probability mass functions
$f_t$ for $t\in \mathbb{N}$. 
Let $G_t$ and $G_{t,j}$ be non-negative real-valued distribution functions with respective probability densities $g_t$ and $g_{t,j}$ for $t, j\in \mathbb{N}$. While it is not necessary but for simplicity, we assume 
  $G_{tj}=G_t$ for any $t,j\in\mathbb{N}$.

We use $\Phi$ and $\phi$ to denote the standard normal distribution and the corresponding density function, respectively.
For a vector $\boldsymbol{\mu}\in\Real^n$ and a $n\times n$ covariance matrix $\boldsymbol{\Sigma}$, we use $\Phi_{\boldsymbol{\mu}, \boldsymbol{\Sigma}}$ to denote the distribution function of multivariate normal distribution with mean $\boldsymbol{\mu}$ and a covariance matrix $\boldsymbol{\Sigma}$, and $\phi_{\boldsymbol{\mu}, \boldsymbol{\Sigma}}$ to denote the corresponding density function.
Now, we are ready to present the multi-year microlevel collective risk model where the Gaussian copula is used to model the dependence.

\begin{model}[The Gaussian copula model for the multi-year microlevel collective risk model]\label{model.2}
Suppose $\boldsymbol{\rho}$ satisfies \eqref{eq.6}.
Then, consider the random vector $\boldsymbol{Z}_t$ whose joint distribution function $H(\boldsymbol{z}_{\tau})$ is given by the following copula model representation
  \begin{equation}\label{eq.11}
  H(\boldsymbol{z}_{(\tau)})=C_{(\boldsymbol{n})}^{(\boldsymbol{\rho})}\left(
  F_1(n_1), G_{1,1}(y_{1,1}), \cdots, G_{1, n_1}(y_{1, n_1}),\cdots
  F_\tau(n_\tau), G_{\tau, 1}(y_{\tau,1}), \cdots, G_{\tau, n_\tau}(y_{\tau, n_\tau})
\right)
  \end{equation}
  where $C_{(\boldsymbol{n})}^{(\boldsymbol{\rho})}$ is a Gaussian copula with correlation matrix $\boldsymbol{\Sigma}_{(\boldsymbol{n})}^{(\boldsymbol{\rho})}$.

\end{model}


From Lemma \ref{thm.1}, the matrix $\boldsymbol{\Sigma}_{(\boldsymbol{n})}^{(\boldsymbol{\rho})}$ is positive definite for any $\boldsymbol{\rho}$ satisfying \eqref{eq.6}. Hence,
$C_{(\boldsymbol{n})}^{(\boldsymbol{\rho})}$ in Model \ref{model.2} is a valid Gaussian copula.
One can see that the estimation of the parameters in \eqref{eq.11} is involved with the calculation of multivariate Gaussian density function which depends on the length of the observer years.
 Let $\boldsymbol{b}= \left( b_1, \cdots, b_\tau \right)$ be vertices where each $b_j$ is equal to either $n_j$ or $n_{j}-1$. Then the corresponding density function of the random vector of $\boldsymbol{Z}_{(\boldsymbol{\tau})}$ at $\boldsymbol{Z}_{(\boldsymbol{\tau})}=\boldsymbol{z}_{(\boldsymbol{\tau})}$ in \eqref{eq.11} is given by
  \begin{equation}\label{eq.12}
  h(\boldsymbol{z}_{(\boldsymbol{\tau})})=
  \sum {\rm sgn}(\boldsymbol{b})    \frac{\partial^{\bar{\boldsymbol{z}}+\tau}}{\partial y_{1,1}\cdots\partial y_{1,n_1},  \partial y_{2,1}\cdots \partial y_{2, n_2}, \cdots, \partial y_{\tau, 1}, \cdots \partial y_{\tau, n_\tau}} H\left( \boldsymbol{z}_{(\boldsymbol{\tau})} \right)
  \end{equation}
  where the sum is taken over all vertices $\boldsymbol{b}$, and ${\rm sgn}(\boldsymbol{b})$ is given by
  \[
  {\rm sgn}(\boldsymbol{b})=
  \begin{cases}
    +1, &\hbox{if $b_j=n_j-1$ for an even number of $j$'s;}\\
    -1, &\hbox{if $b_j=n_j-1$ for an odd number of $j$'s.}
  \end{cases}
  \]
Here, we note that calculation of the density function in \eqref{eq.12} can be difficult due to the following aspects of our model.
\begin{itemize}
  \item Due to the discrete nature of the frequency observations, one can immediately check that the computational complexity in \eqref{eq.12} grows exponentially with $\tau$.
  \item The calculation of each summation in \eqref{eq.12}, which requires a numerical multivariate integration due to the nature of multivariate Gaussian function, can be even difficult especially in high dimensions \citep{Genz2009}
\end{itemize}
However, here we avoid such difficulty by using the following copula representation which is inspired by factor copula representation in \citet{krupskii2013factorcopula}, \citet{nikoloulopoulos2015factor} and \citet{kadhem2019factor}.  For $\boldsymbol{\rho}$ defined in \eqref{eq.6} satisfying \eqref{eq.9},
 we extend the modeling of $\boldsymbol{Z}_{(\boldsymbol{\tau})}$ by including the random effect $R$ so that the joint distribution of $\left(R, \boldsymbol{Z}_{(\boldsymbol{\tau})}\right)$  is given by
  \begin{equation}\label{eq.13}
\begin{aligned}
  &H^*\left(r, \boldsymbol{z}_{(\boldsymbol{\tau})}\right)\\
  &=
  C_{(\boldsymbol{n})}^{(\boldsymbol{\rho}, \boldsymbol{\theta})}\left(\Phi(r),
  F_1(n_1), G_{1,1}(y_{1,1}), \cdots, G_{1, n_1}(y_{1, n_1}),\cdots
  F_\tau(n_\tau), G_{\tau, 1}(y_{\tau,1}), \cdots, G_{\tau, n_\tau}(y_{1, n_\tau})
\right).
\end{aligned}
  \end{equation}
Naturally, by the property of the copula $C_{(\boldsymbol{n})}^{(\boldsymbol{\rho}, \boldsymbol{\theta})}$,  the joint distribution in \eqref{eq.13} implies  the joint distribution function in \eqref{eq.11} in the following sense
\[
H\left(\boldsymbol{z}_{(\boldsymbol{\tau})}\right)=
\lim\limits_{r\rightarrow \infty}H^*\left(r, \boldsymbol{z}_{(\boldsymbol{\tau})}\right),
\]
which further implies that the random vector $\left(R, \boldsymbol{Z}_{(\boldsymbol{\tau})}\right)$ is a natural extension of the random vector $\boldsymbol{Z}_{(\boldsymbol{\tau})}$.
Furthermore, reparameterization in \eqref{eq.6} gives us a well-defined and natural dependence structure with the shared random effect ${R}$ so that claims across multiple years would be independent conditional on ${R}$.
For example, if \eqref{eq.6} holds and $\theta_1=\theta_2=0$, then one can see that $\boldsymbol{Z}_{(\boldsymbol{\tau})}$ are not only conditionally independent but also marginally independent so that $\boldsymbol{Z}_t \perp \boldsymbol{Z}_{t'}$ for all $t \neq t'$.
In addition to \eqref{eq.6}, if $\theta_3=\theta_4=0$, then $\boldsymbol{M}$ is not only block-diagonal, but diagonal, which implies that frequency and severity are independent once the shared random effect $R$ is controlled. In other words,  dependence between frequency and severity are fully explained by the shared random effect  ${R}$. Finally,  if \eqref{eq.6} holds and $\theta_1=\theta_2=\theta_3=\theta_4=0$, then $\boldsymbol{\Sigma}_{(\boldsymbol{n})}^{(\boldsymbol{\rho})}$ is diagonal, which implies our model specification includes the traditional model, which assumes independence among the claims in different years and independence between the frequency and severity.

The following theorem shows us the key idea of our copula representation where
the observed claim $\boldsymbol{Z}_{t}$ for $t=1, \ldots, \tau$ are independent conditional on the random effect $R$, and can be fitted into the special case of Model \ref{model.1}.
In this regard, the copula in \eqref{eq.13} has similar spirit as a factor copula. The corresponding copula of the distribution of $\boldsymbol{Z}_t$ conditional on $R$ is a Gaussian copula which can be represented as 1-factor copula. As a result, the distribution in \eqref{eq.13} have 2-factor copula representation. However, such representation of the model increases the complexity of the notation while provides limited benefit in computational complexity, and hence we do not pursue such representation for the simplicity of the paper.

\begin{theorem}\label{thm.2}
Suppose that $\boldsymbol{\rho}$ satisfies \eqref{eq.6} and joint distribution function $H^*$ of the random vector $\left(R,\boldsymbol{Z}_t\right)$ is given by the factor copula model in \eqref{eq.13}.
Then, we have the following results.
\begin{enumerate}
  \item[i.] The distribution function of $\boldsymbol{Z}_{(\boldsymbol{\tau})}$ can be obtained as in \eqref{eq.11}.
  \item[ii.] The density function of $\boldsymbol{Z}_{(\boldsymbol{\tau})}=\boldsymbol{z}_{(\boldsymbol{\tau})}$ conditional on $R=r$ is given by
  \[
  h^*\left( \boldsymbol{z}_{(\boldsymbol{\tau})}|r\right)
  =\prod\limits_{t=1}^{\tau} h_t^*\left( \boldsymbol{z}_{t}|r\right)
  \]
  where $h_t^*(\cdot|r)$ is the conditional density function of $\boldsymbol{Z}_{t}$ conditional on $R=r$.
  \item[iii.]
  $N_t \perp R \,$ for $t=1, \ldots$ if and only if $\theta_1=0$.
  \item[iv.]
  $\boldsymbol{y}_t \perp R\,$ for $t=1, \ldots$ if and only if $\theta_2=0$.
\end{enumerate}
\end{theorem}
\begin{proof}
The proof of part i is trivial from the property of copula function. The proof of part ii, by the invariance property of the copula under the monotone transformation,
we have that the corresponding copula $C$ of the conditional distribution of random vector $\boldsymbol{Z}_{(\boldsymbol{\tau})}$ conditional on $R=r$ is again
a Gaussian copula. Furthermore, knowing that $C$ is a Gaussian copula,
Theorem \ref{thm.1} shows that $\boldsymbol{Z}_{1}, \cdots, \boldsymbol{Z}_{\tau}$ are independent conditional on $R=r$. The proofs of part iii and iv are immediate from the property of Gaussian copulas.
%
\end{proof}

Based on this result in Theorem \ref{thm.2}, one can obtain the joint density of $(\boldsymbol{Z}_{1}, \cdots, \boldsymbol{Z}_{\tau})$ just with a single dimensional (numerical) integration as the following manner.
\begin{corollary}\label{cor.1}
Consider the random vector $\boldsymbol{Z}_t$ under the settings in Model \ref{model.2}. 
Then, the joint density of  $\boldsymbol{Z}_{(\boldsymbol{\tau})}$  is given as follows:
  \begin{equation*}
\begin{aligned}
  &h(\boldsymbol{z}_{(\boldsymbol{\tau})})\\
  &=\int  \prod\limits_{t=1}^{\tau}\left[
  g_t^{\rm [joint]*}\left(y_{t,1}, \cdots, y_{t, n_t} |r\right)
      \left(
      \Phi\left( \frac{\Phi^{-1}(F(n_t)) -\mu_t}{\sigma_t} \right) - \Phi\left( \frac{\Phi^{-1}(F(n_t-1)) -\mu_t}{\sigma_t} \right)
      \right) \right] \phi(r){\rm d}{r}\\
\end{aligned}
\end{equation*}
where
\begin{equation}\label{eq.14}
\mu_t= \left(\theta_1, \rho_1, \cdots, \rho_1 \right) \left( \boldsymbol{\Sigma}_{(tt)}^{(\boldsymbol{\rho}^*)}
\right)^{-1} \left( r, \Phi^{-1}\left(G(y_{t,1})\right), \cdots, \Phi^{-1}\left(G(y_{t,n_t})\right) \right)^{\mathrm T}
\end{equation}
and
\begin{equation}\label{eq.15}
\left(\sigma_t\right)^2=1-\left(\theta_1, \rho_1, \cdots, \rho_1 \right) \left( \boldsymbol{\Sigma}_{(tt)}^{(\boldsymbol{\rho}^*)}
\right)^{-1}\left(\theta_1, \rho_1, \cdots, \rho_1 \right)^{\mathrm T}
\end{equation}
with $\rho_1^*= \theta_2$ and $\rho_2^*=\rho_2$.
Here, $g_t^{\rm [joint]*}\left(\cdot |r\right)$ is the density function of $\boldsymbol{Y}_{t}$ conditional on $R=r$, and given by
\[
\begin{aligned}
&g_t^{\rm [joint]*}\left(y_{t,1}, \cdots, y_{t,n_t} |r\right)\\
&=\phi_{\boldsymbol{\mu^*}, \boldsymbol{\Sigma}^*}\left(
 \Phi^{-1}\left( G\left(y_{t,1}\right)\right), \cdots, \Phi^{-1}\left( G\left(y_{t,n_t}\right)\right)
 \right)
\prod\limits_{j=1}^{n_t}\frac{g(y_{t,j})}{\phi\left( \Phi^{-1}\left( G(y_{t,j})\right) \right)}
\end{aligned}
\]
where
\[
\boldsymbol{\mu^*}=r\, \theta_2 \boldsymbol{1}_{n_t}\quad\hbox{and}\quad
\boldsymbol{\Sigma}^*= (1-\rho_2) \boldsymbol{I}_{n_t} +  (\rho_2-\theta_2^2) \boldsymbol{J}_{n_t\times n_t}.
\]

\end{corollary}
\begin{proof}[Proof of Corollary \ref{cor.1}]
According to Theorem \ref{thm.2}, we extend $\boldsymbol{Z}_{(\boldsymbol{\tau})}$ to the factor copula model $\left(R, \boldsymbol{Z}_{(\boldsymbol{\tau})}\right)$ having the distribution function in \eqref{eq.13}.
Then, we have
\[
\begin{aligned}
h(\boldsymbol{z}_{(\boldsymbol{\tau})}) &=\int  h^*(\boldsymbol{z}_{\boldsymbol{(\tau)}}|r) \phi(r){\rm d}r \\
&= \int \prod\limits_{t=1}^{\tau}  h_t^*(\boldsymbol{z}_{t}|r) \phi(r){\rm d}r \\
		&= \int \prod\limits_{t=1}^{\tau}  \left[ g_t^{\rm [joint]*}\left(y_{t,1}, \cdots, y_{t, n_t} |r\right) 
\P{N_t=n_t | y_{t,1}, \cdots, y_{t, n_t}, r} \right] \phi(r){\rm d}r
\end{aligned}
\]
where $h^*(\cdot|r)$, and $h_t^*(\cdot|r)$
are the density functions of
$\boldsymbol{Z}_{(\boldsymbol{\tau})}$ and $\boldsymbol{Z}_{t}$, respectively, and
$g_t^{\rm [joint]*}\left(y_{t,1}, \cdots, y_{t, n_t} |r\right)$ is the density function of $\left( Y_{t,1}, \cdots, Y_{t, n_t}\right)$ at $\left( Y_{t,1}, \cdots, Y_{t, n_t}\right)=\left( y_{t,1}, \cdots, Y_{y, n_t}\right)$ conditional on $R=r$.
Here, the second equality is from Theorem \ref{thm.2}, and the final equality is just conditional distribution expression of the joint density function.


Finally, it suffices to show that
$$
\begin{aligned}
\P{n_t | y_{t,1}, \cdots, y_{t, n_t}, r} &= \mathbb{P}(N_t\leq n_t|r,y_{t,1}, \cdots, y_{t,n_t}) - \mathbb{P}(N_t\leq n_t-1|r,y_{t,1}, \cdots, y_{t,n_t})\\
	& =  \Phi\left( \frac{\Phi^{-1}(F(n_t)) -\mu_t}{\sigma_t} \right) - \Phi\left( \frac{\Phi^{-1}(F(n_t-1)) -\mu_t}{\sigma_t}
      \right)
\end{aligned}
$$
and
$$
g_t^{\rm [joint]*}\left(y_{t,1}, \cdots, y_{t,n_t} |r\right)=
\phi_{\boldsymbol{\mu^*}, \boldsymbol{\Sigma}^*}\left(
 \Phi^{-1}\left( G\left(y_{t,1}\right)\right), \cdots, \Phi^{-1}\left( G\left(y_{t,n_t}\right)\right)
 \right)
\prod\limits_{j=1}^{n_t}\frac{g(y_{t,j})}{\phi\left( \Phi^{-1}\left( G(y_{t,j})\right) \right)},
$$
which are proved by Lemmas \ref{lem.1} and \ref{lem.2}, respectively in \ref{appendix.a}.
\end{proof}

\begin{remark}\label{rem.1}
  The model specification in Model \ref{model.2}, and equivalently Model \ref{model.1}, is not casual in the sense that the length or dimension of the observation varies depending on the value of the observation. For example, we have $\boldsymbol{z}_t=(n_t, y_{t,1})$ for $n_t=1$ while
   $\boldsymbol{z}_t=(n_t, y_{t,1}, y_{t,2})$ for $n_t=2$. Hence, the model itself does not seem to be well-defined as it is not even clear how to mathematically define cumulative distribution function or the joint density function. In \ref{appendix.c}, we show how to interpret the density function and corresponding distribution function in Model \ref{model.2} so that they are well-defined.
  Specifically, one can easily check that the copula function $C_{(\boldsymbol{n})}^{(\boldsymbol{\rho})}$ in Model \ref{model.2} satisfies the inheritance property in the similar manner as in \eqref{eq.c2}, which further implies that Model \ref{model.2} can be reformulated as Model \ref{model.4} where the distribution and density functions are well-defined.
  Finally, one can easily show that the distribution and density functions in Model \ref{model.2} are the same as those in Model \ref{model.4} so that they are well-defined. We leave the detailed discussion in \ref{appendix.c}. 
  The discussion on the well-definedness of Model \ref{model.2} but limited to one-year model can be also found in \citet{Oh2019copula}.

\end{remark}

\section{Simulation study} \label{sec.4}

In this section, we conduct a simulation study to investigate the finite sample properties of the parameter estimates and effects of the dependences on them for the proposed method based on Model \ref{model.2}.
We assume one risk class only, where the distribution function $F$ follows Poisson distribution with mean parameter $\lambda_0$ and the distribution function $G$ follows Weibull distribution with mean paramter $\xi_0$ and shape paramter $\nu$.
Here, the parameters for the marginal part of severity are specified as $\xi_0=\exp(8)$, and $\nu=0.7$, which are the same for all scenarios.
The portfolio of policyholders of size $I$ observed for three years ($\tau=3$) are generated from the proposed model under 8 scenarios motivated by the real data analysis in Section \ref{sec.5}.
Table \ref{sim_param} provides the rest of parameter settings and the corresponding correlation coefficients.

\begin{table}[h!]
\caption{Parameter settings of the copula part for each scenario}
\centering
\begin{tabular}{ l *{15}{c}}
 \hline
 &&\multicolumn{5}{l}{Parameter} \\
 \cline{3-7}
 Scenario & I & $\lambda_0 $& $\theta_1$ & $\theta_2$& $\theta_3$& $\theta_4$
 && $\rho_1$ & $\rho_2$& $\rho_3$& $\rho_4$& $\rho_5$\\
 \hline
1 	& 500 & 2 & 0.3& 0.3&	0.5& 0.5&& 0.34& 0.34& 0.09& 0.09& 0.09 \\
2 	& & 		 & 0.3& 0.3&	0  & 0  && 0.09& 0.09& 0.09& 0.09& 0.09  \\
3 	& & 		 & 0.3& 0.7&	0.5& 0.5&& 0.46& 0.74& 0.09& 0.21& 0.49  \\
4 	& & 		 & 0.3& 0.7&	0  & 0  && 0.21& 0.49& 0.09& 0.21& 0.49  \\
5 	& & 		 & 0.7& 0.3&	0.5& 0.5&& 0.46& 0.34& 0.49& 0.21& 0.09  \\
6 	& & 		 & 0.7& 0.3&	0  & 0  && 0.21& 0.09& 0.49& 0.21& 0.09  \\
7 	& & 		 & 0.7& 0.7&	0.5& 0.5&& 0.74& 0.74& 0.49& 0.49& 0.49  \\
8 	& & 		 & 0.7& 0.7&	0  & 0  && 0.49& 0.49& 0.49& 0.49& 0.49  \\
\hline
\end{tabular}
\label{sim_param}
\end{table}

For each scenario, Tables \ref{sim_rb} and \ref{sim_mae} summarize the simulation results from 500 independent Monte Carlo samples, including the relative bias and mean squared error (MSE) of the parameter estimates.
Table \ref{sim_rb} indicates that in all the scenarios, the estimates are close to the true parameters of the proposed model and shows that the relative bias and MSE are small.

\begin{table}[h!]
\caption{Relative bias in $\%$ for all the parameters from the each scenarios}
\centering
\begin{tabular}{ l r r r r r r r r r r r  }
 \hline
 &\multicolumn{5}{l}{RB ($\%$)} \\
 \cline{2-12}
 Scenario & $\lambda_0$ & $\xi_0$ & $\nu$ & $\theta_1$ & $\theta_2$& $\theta_3$& $\theta_4$\\
 \hline
1&-0.21&	-0.06 &	0.17 &	-1.48 &	1.08 &	-0.06 &	-0.27\\
2&-0.26  &		-0.26  &		1.19 &		0.84 &		-0.98 &	-  &	-\\
3&-0.52 &		-0.03 &		-10.77 &		-9.81 &		1.16 &		-1.19 &		0.10\\
4&0.05 &		-0.53 &		1.74 &		6.09 &		0.64 &	- &	-\\
5&0.00 &		-0.08 &		-0.29  &		0.33 &		1.43 &		-1.36 &		0.14\\
6&0.38 &		-0.31 &		18.92  &		-1.70 &		-0.63 &		- &		-\\
7&-0.16 &		-0.03 &		-20.82 &		-19.90 &		0.54 &		-0.63 &		-0.01\\
8&0.14 &		-0.50 &		13.95  &		9.68 &		1.00 &		- &		-\\
\hline \hline
\end{tabular}
\label{sim_rb}
\end{table}

\begin{table}[h!]
\caption{Mean absolute error for all the parameters from the 12 scenarios}
\centering
\begin{tabular}{ l r r r r r r r r r r r  }
 \hline
 &\multicolumn{5}{l}{MSE} \\
 \cline{2-10}
Scenario & $\lambda_0$ & $\xi_0$ & $\nu$  & $\theta_1$ & $\theta_2$& $\theta_3$& $\theta_4$\\
 \hline
1&0.0015 & 	0.0008 & 	0.0023 & 	0.0012 & 	0.0021 & 	0.0007 & 	0.0001\\
2&0.0010 & 	0.0009 &    0.0011 & 	0.0004 & 	0.0001 & 	- & 	-\\
3&0.0014 & 	0.0009 & 	0.0131 & 	0.0039 & 	0.0023 & 	0.0008 & 	0.0001\\
4&0.0015 & 	0.0020 & 	0.0023 & 	0.0012 & 	0.0001 & 	- & 	-\\
5&0.0014 & 	0.0011 & 	0.0038 & 	0.0015 & 	0.0033 & 	0.0008 & 	0.0001\\
6&0.0015 & 	0.0009 & 	0.0023 & 	0.0006 & 	0.0001 & 	- & 	-\\
7&0.0012 & 	0.0011 & 	0.0124 & 	0.0064 & 	0.0025 & 	0.0007 & 	0.0002\\
8&0.0017 & 	0.0019 & 	0.0069 & 	0.0024 & 	0.0001 & 	- & 	-\\
\hline\hline
\end{tabular}
\label{sim_mae}
\end{table}

\section{Empirical application} \label{sec.5}

In this section, we now calibrate the proposed model to a real auto insurance dataset,
to examine dependence structure
(i) between frequency and severity within a year,
(ii) among two distinct severities within a year,
(iii) among frequencies across years,
(iv) between frequency and severity in different years,
and (v) between two severities in different years.

\subsection{Data}

For this empirical investigation, we employ a dataset from a general insurer in Singapore, which consists of a portfolio of personal automobile insurance policies with comprehensive coverages. The dataset has been obtained from General Insurance Association of Singapore, a trade association with representations from all the general insurance companies in Singapore. The claims experience observed from this dataset is longitudinal over a period of six years, from 1995 to 2000, and has 17,452 unique policyholders.  Among the observations, we randomly sample 5000 policyholders.
To calibrate the models, the observations for the first five years, 1995-1999 is used as in-sample, or training data, and in order to examine the performance of the model, we use the last year 2000 as the hold-out sample, or test data.

The dataset also contains a set of predictors that could further explain additional variation in the number of claims and the claim amounts.  To summarize the variables observed, we have three categorical variables and one continuous variable: the gender with two levels (male and female), insured's age (Age) with four levelss including age 1 $\in (0,25]$, age 2 $\in (25,35]$, age 3 $\in (35,65]$, and age 4 $\in (65, \infty]$, vehicle age (VehAge) with four levels including vehicle age 1 $\in [0,1]$, vehicle age 2 $\in (1,5]$, vehicle age 3 $\in (5,10]$, and vehicle age 4 $\in (10,\infty]$, and vehicle's capacity expressed in log scale (logVehCapa).

Table \ref{tab.x} summarizes the description and simple statistics of these predictor variables which represent the risk characteristics of policyholders: Gender, Age, VehAge,  and logVehCapa. In Singapore, as observed in this table, there is a disproportionate distribution by gender with more male than female drivers. When we the distribution of drivers by age, it is also not surprising to find fewer percentage of younger drivers, unlike that in other developed countries. The primary reason for this is the extremely expensive cost of owning and maintaining a car, in addition to the efficiency of the use of public transportation. During the period of observation, it is highly discourage to own a car for more than 10 years, and this reflected in this distribution.

Furthermore, a summary of the claim frequency over the years 1995 to 1999 is given in Table \ref{tab.n} and the average claim amount categorized by frequency and year is given in Table \ref{tab.y}. This table suggests that the claims size appears to be unstable over time. We adjust the values of the individual severities, in order to satisfy that the average of individual severity over each year is the same as the average over the 2000 hold-out sample data with 4,659 observations.

\begin{table}[h!t!]
\centering
\caption{Observable policy characteristics used as covariates} \label{tab.x}
\begin{tabular}{l|l r r r r r r r }
\hline \hline
Categorical & \multirow{2}{*}{Description} &&  \multicolumn{3}{r}{\multirow{2}{*}{Proportions}} \\
variables &  &  &  &   \\
\hline
Gender   & Insured's sex:   \\
		& \qquad\qquad  Male = 1  	&& \multicolumn{3}{r}{80.03$\%$} \\
 		& \qquad\qquad  Female = 0 	&& \multicolumn{3}{r}{19.97$\%$} \\
\hline
Age &  	The policyholder's issue age : \\
	& \qquad\qquad Age $\in (0,\,\,\, 25] = 1 $ 	&&   \multicolumn{3}{r}{0.49$\%$} \\ 	
	& \qquad\qquad Age $\in (25, 35] = 2 $ 		&&   \multicolumn{3}{r}{21.68$\%$} \\ 	
	& \qquad\qquad Age $\in (35, 65] = 3 $ 		&&  \multicolumn{3}{r}{76.81$\%$}  \\ 	
	& \qquad\qquad Age $\in (65, \infty] = 4 $ 	&&  \multicolumn{3}{r}{1.03$\%$}  \\ 	
\hline
VehAge &  	Age of vehicle in years : \\
	& \qquad\qquad VehAge $\in [0,\,\quad 1] = 1 $ 	&&  \multicolumn{3}{r}{12.45$\%$} \\ 	
	& \qquad\qquad VehAge $\in (1,\quad 5] = 2 $ 		&&  \multicolumn{3}{r}{57.30$\%$} \\ 	
	& \qquad\qquad VehAge $\in (5, \,\,\,10] = 3 $ 	&&  \multicolumn{3}{r}{29.99$\%$}  \\ 	
	& \qquad\qquad VehAge $\in (10, \infty] = 4 $ 	&&  \multicolumn{3}{r}{0.25$\%$}  \\ 	
\hline	
 Continuous & && \multirow{2}{*}{Min} & \multirow{2}{*}{Mean} & \multirow{2}{*}{Max} \\
 variables \\
\hline
logVehCapa	& Insured vehicle's capacity in cc  &&  6.49 	& 7.19 	& 8.82  \\
\hline \hline
\end{tabular}
\end{table}

\begin{table}[h!t!]
\centering
\caption{Number of observations by frequency and year} \label{tab.n}
\begin{tabular}{l r r r r r r r r r r r }
\hline \hline
&& \multicolumn{7}{l}{Train} && \multicolumn{2}{l}{Test}\\ \cline{3-9} \cline{11-12}
Frequency && 1995 	& 1996 	& 1997 	& 1998 	& 1999	& Count	& $\%$ of Total && 2000 & $\%$ of Total\\
\hline
0	 &&  3103		& 3291 	& 2501 	& 2036 	& 1751 	& 12682	& 91.05 && 	1360 & 92.39 \\
1	 && 	 232 	&   212  &  266	&  214	&  219  	&  1143	&  8.21 &&	 104 & 7.07\\
2 	 &&   17  	&    8   &   20   &   24  	&   18  	&   87	&  0.62 &&	  8  & 0.54\\
3 	 &&    2  	&    1   &   2 	&   2 	& 	4	& 	11	&  0.08 && 	   0 & 0\\
\hline
Count && 3354	& 3512 	& 2789 	& 2276 	& 1992	& 13923 	& 100 && 	4659 & 100 \\
\hline \hline
\end{tabular}
\end{table}

\begin{table}[h!t!]
\centering
\caption{Average severity by frequency and year} \label{tab.y}
\begin{tabular}{l r r r r r r r r r r }
\hline \hline
&& \multicolumn{6}{l}{Train} && \multicolumn{2}{l}{Test}\\ \cline{3-8} \cline{10-10}
Frequency && 1995 	& 1996 	& 1997 	& 1998 	& 1999	& Avg. severity && 2000\\
\hline
1	 && 	 4742 	&   4530  &  4567	&  5440	&  3895  	&  4630 && 4557\\
2 	 &&   6319  	&   3633  &  3629 &  3781  	&  3644  	&  4200 && 2950\\
3 	 &&   2630 	&   1687  &  4991 	&  6015 	&  3065	&  3747 && -\\
4   	 && 		- 	& 	- 	 &	 - 	&  	- 	&   - 	&  -    && - \\
\hline
Avg. severity && 4892	&  4431 	& 4455 	& 5156 	& 3824	& 4553 && 4046\\
\hline \hline
\end{tabular}
\end{table}

\subsection{Estimation result}

For the data analysis, we consider the model with regression setting described in Corollary \ref{cor.1}. We assume the distribution function, $F$, follows a Poisson distribution with mean parameter, $\lambda$, for the frequency, and the distribution function, $G$, follows a Weibull distribution with mean parameter, $\xi$, and shape parameter, $\nu$, for the severity component.  With a log link function, we therefore have
\[
\lambda=\exp(\boldsymbol{x}\boldsymbol{\beta}), \quad\hbox{and}\quad
\xi=\exp(\boldsymbol{w}\boldsymbol{\gamma}),
\]
where $\boldsymbol{x}$ and $\boldsymbol{w}$ are the vectors of model matrices for each policyholder
\footnote{In this example, $\boldsymbol{x}$ includes Gender, Age, and VehAge, and
 $\boldsymbol{w}$ includes VehCapa, Age, and VehAge.}, and $\boldsymbol{\beta}$ and $\boldsymbol{\gamma}$ are the corresponding parameters for the frequency and severity, respectively. Hence, in this data analysis, we consider following parameters: $(\boldsymbol{\beta},\boldsymbol{\gamma},\nu, \theta_1, \theta_2,\theta_3,\theta_4)$.

Table \ref{res.est} presents the summary statistics for the model estimation results. This table provides details of the estimated parameters for the frequency part, the severity part, as well as the copula part. There are four measures detailed in this table: estimates (est), standard errors (std.error), t statistics (t), and corresponding p-values. Note that the asterisk sign (*) indicates that the estimate is significant at 0.05 level. From the table, the results are as expected. For example, despite the disproportionate percentage, male drivers are expected to incur more accidents than female drivers. When analyzed by age, broadly speaking, both frequency and severity tend to decrease with age. Elderly drivers, for example, have fewer number of accidents with smaller average costs per accident than drivers less than 25 years old.

\begin{table}[ht!]
\caption{Estimation result} \label{res.est}
\centering
\begin{tabular}{ l r r r r r r r  }
 \hline
parameter& est & std.error & t & p-value & \\ \hline
 \multicolumn{4}{l}{ {\bf Frequency part}} \\
(Intercept)	&	-2.237 	&	0.289 	&	-7.749 	&	$<$.0001	&	*	\\
Gender		&	0.125 	&	0.067 	&	1.865 	&	0.0623 	&		\\
VehAge2		&	0.048 	&	0.101 	&	0.477 	&	0.6336 	&		\\
VehAge3		&	-0.146 	&	0.109 	&	-1.339 	&	0.1806 	&		\\
VehAge4		&	0.835 	&	0.443 	&	1.886 	&	0.0594 	&		\\
Age2			&	0.323 	&	0.270 	&	1.197 	&	0.2313 	&		\\
Age3			&	0.156 	&	0.268 	&	0.583 	&	0.5601 	&		\\
Age4			&	-0.569 	&	0.460 	&	-1.237 	&	0.2161 	&		\\
 \multicolumn{4}{l}{{\bf Severity part} } \\
(Intercept)	&	3.889 	&	0.938 	&	4.148 	&	$<$.0001	&	*	\\
logVehCapa	&	0.700 	&	0.108 	&	6.468 	&	$<$.0001	&	*	\\
VehAge2		&	-0.010 	&	0.097 	&	-0.099 	&	0.9212 	&		\\
VehAge3		&	-0.060 	&	0.110 	&	-0.547 	&	0.5843 	&		\\
VehAge4		&	-0.624 	&	0.565 	&	-1.106 	&	0.2690 	&		\\
Age2			&	-1.092 	&	0.478 	&	-2.284 	&	0.0224 	&	*	\\
Age3			&	-0.976 	&	0.475 	&	-2.055 	&	0.0400 	&	*	\\
Age4			&	-0.969 	&	0.668 	&	-1.451 	&	0.1468 	&		\\
$\nu$		&	0.802 	&	0.045 	&	17.910 	&	$<$.0001	&	*	\\
\multicolumn{4}{l}{{\bf Copula part}} \\
$\theta_1$	&	0.263 	&	0.048 	&	5.509 	&	$<$.0001	&	*	\\
$\theta_2$	&	0.057 	&	0.071 	&	0.795 	&	0.4266 	&		\\
$\theta_3$	&	0.409 	&	0.138 	&	2.967 	&	0.0030 	&	*	\\
$\theta_4$	&	0.445 	&	0.133 	&	3.341 	&	0.0008 	&	*	\\
\hline
\end{tabular}
\end{table}

In terms of understanding the presence of dependence, Table \ref{res.est} also summarizes estimates of the four copula parameters of dependence as described by $\theta_i$, for $i=1,2,3,4$, and the estimates are not all significantly nonzero at the 5$\%$ level. For the interpretation of copula parameters in Table \ref{res.est}, one can recall the following meaning of $\theta_1, \theta_2, \theta_3$, and $\theta_4$:
	\begin{itemize}
		\item $\theta_1$: dependence parameter between the common random effect ${R}$ and the frequency for every year,
		\item $\theta_2$: dependence parameter between the common random effect ${R}$ and each severity for every year,
		\item $\theta_3$, $\theta_4$: dependence parameters between a frequency and each severity within a year not explained by ${R}$.
	\end{itemize}
Thus, according to the estimation results which shows that only $\theta_1$ and $\theta_2$ are significantly different from zero, we can claim presence of both types of dependence; temporal dependence of claim frequencies and severities as well as dependence between the frequency and severity can be explained by common random effect ${R}$. On the other hand, there is weak evidence of dependence between a frequency and each frequency within a year not explained by ${R}$.

While the values of $\theta$ tells us the relationship between the common random effects and claims, one can directly quantify the magnitude of dependence among the claims by observing the estimated values of $\rho$'s. According to the model specification in \eqref{eq.6}, the estimates of dependence parameters, $\rho_1, \, \rho_2, \, \rho_3, \, \rho_4,$ and $\rho_5$ are calculated from the estimates of $\theta_1, \, \theta_2, \, \theta_3$, and $\theta_4$ as in \eqref{eq.6}, and their standard errors are obtained using delta method.
Table \ref{rho:est} summarizes the derived estimates together with the resulting standard errors of $\rho_1, \, \rho_2, \, \rho_3, \, \rho_4,$ and $\rho_5$.

\begin{table}[ht!]
\caption{Derived estimates and standard errors of $\rho$'s} \label{rho:est}
\centering
\begin{tabular}{ l r r r r r r r  }
\hline parameter& est & std.error & t & p-value & \\ \hline
$\rho_1$	&	0.1968	&	0.074	&	2.655	&	0.0079	&	*	\\
$\rho_2$	&	0.2015	&	0.119	&	1.688	&	0.0915	&		\\
$\rho_3$	&	0.0690	&	0.025	&	2.754	&	0.0059	&	*	\\
$\rho_4$	&	0.0149	&	0.019	&	0.780	&	0.4355	&		\\
$\rho_5$	&	0.0032	&	0.008	&	0.398	&	0.6909	&		\\
\hline
\end{tabular}
\end{table}

It is interesting to observe that there is now a clearer evidence of all types of dependencies in our multi-year microlevel model. For example, $\rho_1$ describes correlation between a frequency and a severity within a single year, and results provide strong evidence of a positive dependence. The estimate for $\rho_1$ is $0.1968$ with a standard error of $0.074$, which leads to a very small p-value indicating significantly different from zero. Using the results from (\ref{eq.6}), despite the non-significance of $\theta_3$ and $\theta_4$ directly drawn from the estimated model, there is a clear inherent dependence driven by the shared random effect through the interplay with $\theta_1$ and $\theta_2$. A similar argument can be said of the other $\rho$'s.

\subsection{ Validation }

For validation of the proposed model in terms of the individual loss prediction for the 1,472 policyholders in the hold-out sample, we compare the following four models: full model, nested model 1 with $\theta_3=\theta_4=0$, the nested model 2 with $\theta_1=\theta_2=0$, and the nested model 3 with $\theta_1=\theta_2=\theta_3=\theta_4=0$. We measure the quality of prediction as mean squared error (MSE) of average of individual loss prediction over 5,000 Monte Carlo simulations under the estimation result from each model. We also use other measures such as  root-mean-square deviation (RMSE), mean absolute error (MAE), and the Gini index in \citet{frees2011summarizing, frees2014insurance}.
 For example, MSE for full model is given as
\[
\widehat{{\rm MSE}}^{\rm [full]} = \frac{1}{1472} \sum_{i=1}^{1472} {\left( S_i - \hat{S}_i^{\rm [full]} \right)^2},
\]
where $S_i$ is the observed aggregate loss for the $i$-th policyholder in 2000, and
$\hat{S}_i^{\rm[full]}$ is the average of predicted aggregate loss for the $i$-th policyholder over 5,000 MC samples from based on the full model.
The results are shown in Table \ref{mse}.
In the table,  full model shows the best performance in terms of MSE and RMSE, and nested model 1 shows best performance in terms of MAE. On the other hand, nested model 2 shows the best performance in terms of the Gini index.

\begin{table}[ht!]
\caption{Means squared error } \label{mse}
\centering
\begin{tabular}{p{2cm} p{2cm} p{2cm} p{2cm} p{2cm}  }
\hline
& Full & Nested 1  & Nested 2  & Nested 3  \\ \hline
RMSE & 2445.409 &  2448.719 & 2445.519  & 2448.276 \\
MSE &  5980026 &  5996227 & 5980564  & 5994053 \\
MAE & 596.87 & 524.2761 & 605.0203 & 530.9766 \\
Gini & 28.535 & 30.478& 27.560& 29.547\\
\hline
\end{tabular}
\end{table}

\section{Final remarks} \label{sec.6}

This article focuses on the development of a multi-year microlevel collective risk model which accounts for a flexible dependence structure for claim frequencies and claim severities. The common theme in the literature is a framework that regards dependence between claim frequency and the average severity. Our motivating example demonstrates that for these types of dependence models, the copula structure can be constrained. Here, we also show that it is even difficult to arrive at the naive assumption of independence among severities.

In our multi-year microlevel collective risk model, we develop a shared random effects framework that captures various relevant types of dependence between claim frequencies and claim severities over multiple years. The shared random effect parameter induces several forms of dependence; it has similar structure to a one-factor copula model previously studied. Our proposed scheme has the advantages of not only ease of computation but also the capacity to draw intuitive interpretation to the results. Furthermore, it covers other types of dependent frequency and severity models that have previously been studied. One can see that both one-year dependent compound risk model and traditional independent compound risk model are special cases of our proposed model, where $\theta_1=\theta_2=0$ and $\theta_1=\theta_2=\theta_3=\theta_4=0$, respectively. In the paper, we additionally provide an efficient way to obtain the joint density of multi-year claim required without heavy numerical integration.

We calibrated our proposed with a dataset from a Singapore automobile insurance company, which contains policy characteristics and microlevel claims information for multiple years. The estimation results show us that all five types of correlations considered in a multi-year microlevel collective risk model are statistically significant. We note that the driving force for the dependencies originates from the shared random effect parameter. On top of that, out-of-sample validation results with the proposed model show us that it can be helpful to consider various types of dependence to increase the prediction performances.

\section*{Acknowledgements}
 Rosy Oh was supported by Basic Science Research Program through the National Research Foundation of Korea (NRF) funded by the Ministry of Education (2019R1A6A1A11051177 and 2020R1I1A1A01067376).
Himchan Jeong was supported by James C. Hickman Doctoral Stipend funded by the Society of Actuaries (SOA).
Emiliano A. Valdez was supported by CAE Research Grant on Applying Data Mining Techniques in Actuarial Science funded by the Society of Actuaries (SOA).
Jae Youn Ahn was supported by a National Research Foundation of Korea (NRF) grant funded by the Korean Government (2020R1F1A1A01061202).

\newpage

\appendix
\numberwithin{equation}{section}
\renewcommand{\theequation}{A.\arabic{equation}}
\section{Proof of Lemmas 1 and 2}\label{appendix.a}

\begin{lemma}\label{lem.1}
If $\left( R, \boldsymbol{Z}_{(\boldsymbol{\tau})} \right)$ follows \eqref{eq.13}, then the probability of $\Phi^{-1}(F_t(N_t))$ conditional on
\[
\left( R, Y_{t,1}, \ldots, Y_{t,n_t}\right)=\left( r, y_{t,1}, \ldots, y_{t,n_t}\right)
\]
 is given by
\begin{equation*} 
\mathbb{P}(N_t\leq n|r,y_{t,1}, \cdots, y_{t, n_t}) = \Phi \left( \frac{\Phi^{-1}(F_t(n))-\mu_t}{\sigma_t} \right)
\end{equation*}
where
\[
\mu_t= \left(\theta_1, \rho_1, \cdots, \rho_1 \right) \left( \boldsymbol{\Sigma}_{(tt)}^{(\boldsymbol{\rho}^*)}\right)^{-1} \left( r, \Phi^{-1}\left(G(y_{t,1})\right), \cdots, \Phi^{-1}\left(G(y_{t,n_t})\right) \right)^{\mathrm T}
\]
and
\[
\left(\sigma_t\right)^2=1-\left(\theta_1, \rho_1, \cdots, \rho_1 \right) \left( \boldsymbol{\Sigma}_{(tt)}^{(\boldsymbol{\rho}^*)}\right)^{-1}\left(\theta_1, \rho_1, \cdots, \rho_1 \right)^{\mathrm T}
\]
with $\rho_1^*=\theta_2, \, \rho_2^*=\rho_2$.
\end{lemma}

\begin{proof}
By the inherited property of the Gaussian copula, it is easy to see that
\begin{equation}\label{eq.a1}
\begin{aligned}
\P{ \Phi^{-1}(F_t(N_t))\le x_0, R\le r, \Phi^{-1}(G_{t,1}(Y_{t,1}))\le x_1, \cdots, \Phi^{-1}(G_{t, n_t}(Y_{t,n_t}))\le x_{n_t}  }\\
=\Phi_{{\bf 0}_{n_t+2}, \boldsymbol{\Sigma}_t}\left(x_0, r, x_1, \cdots, x_{n_t} \right)
\end{aligned}
\end{equation}
where
$$
[\Sigma_t]_{\ell m} =\begin{cases}
  1		, & \ell=m;\\
\theta_1  , & \ell + m = 3;\\
\rho_1     , & 1 = \ell, m >2 \text{ or } \ell > 2, m =1;\\
\theta_2  , & 2=\ell < m \text{ or } \ell > m =2;\\
\rho_2	, & \hbox{elsewhere};
\end{cases} =
\left(
  \begin{array}{c:c:cccc}
    1 & \theta_1 & \rho_1 & \cdots & \cdots &\rho_1 \\
             \hdashline
    \theta_1 &  1 &  \theta_2 & \cdots & \cdots &\theta_2 \\
             \hdashline
    \rho_1 & \theta_2 & 1 & \rho_2 & \cdots &\rho_2  \\
    \vdots & \vdots & \rho_2 & \ddots & \ddots &\vdots  \\
    \vdots & \vdots & \vdots & \ddots & 1 &\rho_2  \\
    \rho_1 & \theta_2 & \rho_2 & \cdots & \rho_2 & 1  \\
  \end{array}
\right)
$$
Furthermore, \eqref{eq.a1} implies
 \[
 \P{\Phi^{-1}(F_t(N_t))\le x_0 | R=r, Y_{t,1}=y_{t,1}, \ldots, Y_{t,n_t}=y_{t,n_t}}=\Phi\left( \frac{x_0-\mu_t}{\sigma_t}\right)
 \]
%
where
\[
\mu_t= \left(\theta_1, \rho_1, \cdots, \rho_1 \right) \left( \boldsymbol{\Sigma}_{(tt)}^{(\boldsymbol{\rho}^*)}
\right)^{-1} \left( r, \Phi^{-1}\left(G(y_{t,1})\right), \cdots, \Phi^{-1}\left(G(y_{t,n_t})\right) \right)^{\mathrm T}
\]
and
\[
\left(\sigma_t\right)^2=1-\left(\theta_1, \rho_1, \cdots, \rho_1 \right) \left( \boldsymbol{\Sigma}_{(tt)}^{(\boldsymbol{\rho}^*)}
\right)^{-1}\left(\theta_1, \rho_1, \cdots, \rho_1 \right)^{\mathrm T}
\]
with $\rho_1^*=\theta_2, \, \rho_2^*=\rho_2$.

\end{proof}

\begin{lemma}\label{lem.2}
Consider the settings in \eqref{eq.13}. Then, the density function of $\left( R, \boldsymbol{Z}_{t} \right)$ conditional on $R=r$ is given by
  \begin{equation*}
g^{\rm [joint]*}_t\left(y_{t,1}, \cdots, y_{t, n_t} |r\right)=
\phi_{\boldsymbol{\mu^*}, \boldsymbol{\Sigma}^*}\left(
 \Phi^{-1}\left( G\left(y_{t,1}\right)\right), \cdots, \Phi^{-1}\left( G\left(y_{t,n_t}\right)\right)
 \right)
\prod\limits_{j=1}^{n_t}\frac{g(y_{t,j})}{\phi\left( \Phi^{-1}\left( G(y_{t,j})\right) \right)},
\end{equation*}
where
\[
\mu^*=r\, \theta_2 \boldsymbol{1}_{n_t}\quad\hbox{and}\quad
\boldsymbol{\Sigma}^*=(1-\rho_2) \boldsymbol{I}_{n_t} +  (\rho_2-\theta_2^2) \boldsymbol{J}_{n_t\times n_t}.
\]
\end{lemma}

\begin{proof}
By the inherited property of the Gaussian copula, it is easy to see that
\[
\left( R, Y_{t,1}, \ldots, Y_{t,n_t} \right) \sim C_{(tt)}^{(\boldsymbol{\rho}^*)}
  \left( \Phi,   G_{t,1}, \cdots, G_{t, n_t} \right)
\]
 with $\rho_1^*=\theta_2, \, \rho_2^*=\rho_2$, where $C_{(tt)}^{(\boldsymbol{\rho}^*)}$ is a Gaussian copula with correlation matrix $\boldsymbol{\Sigma}_{(tt)}^{(\boldsymbol{\rho}^*)}$.
As a result, we have that, conditional on $R=r$,  the random vector
$\Phi^{-1}(G_{t,1}(y_{t,1})), \cdots, \Phi^{-1}(G_{t, n_t}(y_{t,n_t})) $ follows a multivariate normal distribution, with mean vector given as
$$
(\theta_2, \ldots, \theta_2)^{\mathrm{T}} \cdot {I_1}^{-1} \cdot r = r \theta_2 \boldsymbol{1}_{n_t},
$$
and covariance matrix given as
$$
(1-\rho_2) \boldsymbol{I}_{n_t} +  \rho_2 \boldsymbol{J}_{n_t\times n_t}  - (\theta_2, \ldots, \theta_2) \cdot (\theta_2, \ldots, \theta_2)^{\mathrm{T}} = (1-\rho_2) \boldsymbol{I}_{n_t} +  (\rho_2-\theta_2^2) \boldsymbol{J}_{n_t\times n_t}.
$$

Therefore, we have
$$
\P{Y_{t,1}\le y_{t,1}, \cdots, Y_{t, n_t}\le y_{t, n_t} |r}=
\Phi_{\mu^*, \boldsymbol{\Sigma}^*}\left(
 \Phi^{-1}\left( G\left(y_{t,1}\right)\right), \cdots, \Phi^{-1}\left( G\left(y_{t,n_t}\right)\right)
 \right)
$$
with the following corresponding density function
$$
g^*_{t}\left(y_{t,1}, \cdots, y_{t, n_t} |r\right)=
\phi_{\mu^*, \boldsymbol{\Sigma}^*}\left(
 \Phi^{-1}\left( G\left(y_{t,1}\right)\right), \cdots, \Phi^{-1}\left( G\left(y_{t,n_t}\right)\right)
 \right)
\prod\limits_{j=1}^{n_t}\frac{g(y_{t,j})}{\phi\left( \Phi^{-1}\left( G(y_{t,j})\right) \right)}.
$$
\end{proof}

\setcounter{equation}{0}
\renewcommand{\theequation}{B.\arabic{equation}}
\section{Multi-year microlevel collective risk model with $t$ copulas}\label{appendix.b}

The microlevel collective risk model with the Gaussian copula in Model \ref{model.2} can be naturally extended to $t$ copula based model as the following model shows. Here, we use $F_t$, $G_t$, $G_{t,j}$ considered in Section \ref{sec.3.2}, and assume $G_{t,j}=G_t$ for any $t,j\in\mathbb{N}$ for simplicity.

\begin{model}[The $t$ copula model for the multi-year microlevel collective risk model]\label{model.3}
Suppose $\boldsymbol{\rho}$ satisfies \eqref{eq.6}.
Then, consider the random vector $\boldsymbol{Z}_t$ whose joint distribution function $H(\boldsymbol{z}_{\tau})$ is given by the following copula model representation
  \begin{equation}\label{eq.b1}
  H(\boldsymbol{z}_{(\tau)})=C_{(\boldsymbol{n})}^{{\nu}, (\boldsymbol{\rho})}\left(
  F_1(n_1), G_{1,1}(y_{1,1}), \cdots, G_{1, n_1}(y_{1, n_1}),\cdots
  F_\tau(n_\tau), G_{\tau, 1}(y_{\tau,1}), \cdots, G_{\tau, n_\tau}(y_{\tau, n_\tau})
\right)
  \end{equation}
  where $C_{(\boldsymbol{n})}^{{\nu}, \boldsymbol{\rho}}$ is a $t$ copula with scale matrix $\boldsymbol{\Sigma}_{(\boldsymbol{n})}^{(\boldsymbol{\rho})}$ and degree of freedom $\nu$.

\end{model}

Following the similar idea in Section \ref{sec.3.2},  we have the following result.


\begin{corollary}\label{cor.2}
Consider the random vector $  \boldsymbol{Z}_{(\boldsymbol{\tau})}$ whose distribution function is defined in \eqref{eq.b1}. Then, we have the following density function of $  \boldsymbol{Z}_{(\boldsymbol{\tau})}$ at $  \boldsymbol{Z}_{(\boldsymbol{\tau})}=\boldsymbol{z}_{(\boldsymbol{\tau})}$
  \begin{equation}\label{eq.b2}
\begin{aligned}
  &h(\boldsymbol{z}_{(\boldsymbol{\tau})})\\
  &=\int\int  \prod\limits_{t=1}^{\tau}\Bigg[g_{t|R, W}^*\left(y_{t,1}, \cdots, y_{t, n_t} |r, w\right)\\
  &\quad\quad
      \left(
      \Phi\left( \frac{\Phi^{-1}(F(n_t)) -\mu_t}{\sigma_t} \right) - \Phi\left( \frac{\Phi^{-1}(F(n_t-1)) -\mu_t}{\sigma_t} \right)
      \right)\Bigg] \phi(r){\nu} f_{\nu}^{\rm [Chi]}(w \, {\nu}) \,{\rm d}{r}{\rm d}{w}\\
\end{aligned}
\end{equation}
where $f_{\nu}^{\rm [Chi]}$ is a density function of chi-squared distribution with ${\nu}$ degrees of freedom, and
\[
\mu_t= \left(\theta_1, \rho_1, \cdots, \rho_1 \right) \left( \boldsymbol{\Sigma}_{(tt)}^{(\boldsymbol{\rho}^*)}
\right)^{-1} \left( r, \Phi^{-1}\left(G(y_{t,1})\right), \cdots, \Phi^{-1}\left(G(y_{t,n_t})\right) \right)^{\mathrm T}
\]
and
\[
\left(\sigma_t\right)^2=\frac{1-\left(\theta_1, \rho_1, \cdots, \rho_1 \right) \left( \boldsymbol{\Sigma}_{(tt)}^{(\boldsymbol{\rho}^*)}
\right)^{-1}\left(\theta_1, \rho_1, \cdots, \rho_1 \right)^{\mathrm T}}{w}
\]
with $\rho_1^*= \theta_2$ and $\rho_2^*=\rho_2$.
Here, $g_{t|R, W}^*\left(\cdot |r\right)$ is the density function of $\boldsymbol{Y}_t$ conditional on $R=r$ and $W=w$, and given by
\[
\begin{aligned}
&g_{t|R, W}^*\left(y_{t,1}, \cdots, y_{t, n_t} |r, w\right)\\
&=\phi_{\boldsymbol{\mu^*}, \boldsymbol{\Sigma}^*}\left(
 \Phi^{-1}\left( G\left(y_{t,1}\right)\right), \cdots, \Phi^{-1}\left( G\left(y_{t,n_t}\right)\right)
 \right)
\prod\limits_{j=1}^{n_t}\frac{g(y_{t,j})}{\phi\left( \Phi^{-1}\left( G(y_{t,j})\right) \right)}
\end{aligned}
\]
where
\[
\boldsymbol{\mu^*}=r\, \theta_2 \boldsymbol{1}_{n_t}\quad\hbox{and}\quad
\boldsymbol{\Sigma}^*=\frac{(1-\rho_2) \boldsymbol{I}_{n_t} +  (\rho_2-\theta_2^2) \boldsymbol{J}_{n_t\times n_t}}{w}.
\]
\end{corollary}

\begin{proof}[Proof of Corollary \ref{cor.2}]
Knowing that multivariate $t$-distribution with the degree of freedom $\nu$ can be represented as a multivariate normal distribution conditional on the latent variable $W=w$ whose density function at $W=w$ is given by ${\nu} f_{\nu}^{\rm [Chi]}(w \, {\nu})$, the proof follows immediately from Corollary \ref{cor.1}.
\end{proof}
Note that Corollary \ref{cor.1} is a special case of Corollary \ref{cor.2} when $\nu = \infty$.

\setcounter{equation}{0}
\renewcommand{\theequation}{C.\arabic{equation}}
\section{Mathematical Justification of Model \ref{model.1}}\label{appendix.c}
   As briefly discussed in Remark \ref{rem.1}, Model \ref{model.1} is not casual in the sense that the length or dimension of the observation varies depending on the value of the observation.
   One solution to detour the difficulty from the varying dimension of the observation is that we may assume the infinite number of severities regardless of the value of the frequency $n_t$. Specifically, we define
   \[
   \boldsymbol{Z}_t(\boldsymbol{k}_t):=\left(N_1, \boldsymbol{Y}_{1}(k_1), \cdots, N_t, \boldsymbol{Y}_t(k_t) \right)
   \]
   for any $\boldsymbol{k}_t:=(k_1, \cdots, k_t)\in \mathbb{N}_0^{t}$. Then, $\boldsymbol{y}_{t}=\boldsymbol{y}_{t}(n_t)$ can be understood as the observation where we only observe
   first $n_t$ severities among the infinite number of severities.
   Then, Model \ref{model.1} can be reformulated as follows.

   \begin{model}[Revision of Model \ref{model.1}]\label{model.4}
We repeatedly define the following random effect model for all possible values of
\begin{equation}\label{eq.c1}
k_t\in \mathbb{N}_0, \quad t=1, \cdots
\end{equation}
 where the joint distribution between observations and the random effect model is presented with the copula model with parts i and iv are the same as in Model \ref{model.1}.
\begin{enumerate}
  \item[ii.] Conditional on $R=r$, we have that
  $\boldsymbol{Z}_t(k_t)$ for $t=1, \cdots$ are independent observations whose distribution function is given by
      \[
      H_t^\#\left( \boldsymbol{z}_t(k_t)|r\right):=C_{(\theta_3, \theta_4)}\left( F_t(n_t|r),  G_t(y_{t,1}|r), \ldots, G_t(y_{t,k_t}|r) \right).
      \]
      As a result, we have the following distribution function of $\boldsymbol{Z}_{(\tau)}(\boldsymbol{k}_\tau)$
      \[
      H^\#(\boldsymbol{z}_{(\tau)}(\boldsymbol{k}_\tau)):=\int \prod\limits_{t=1}^{\tau} H_t^\#\left( \boldsymbol{z}_t(k_t)|r\right) \pi(r) {\rm d}r.
      \]

      \item[iii.] The parameters $\theta_3$ and $\theta_4$ of the copula $C_{(\theta_3, \theta_4)}$ controls the independence between the frequency and severities and independence among individual severities, respectively, within a year so that we have 
$$
{h_t^\#{(\boldsymbol{z}_t(k_t)|r)} } = f_t(n_t|r) g_t^{\rm [joint]\#}(\boldsymbol{y}_t(k_t)|r) \ \text{ if and only if } \ \theta_3=0,
$$
where
$$
g_t^{\rm [joint]\#}{( \boldsymbol{y}_t(k_t)|r)} =\prod\limits_{j=1}^{k_t} g_t{(y_{t,j}|r)}  \ \text{ if and only if } \ \theta_4=0,
$$
where $g_t^{\rm [joint]\#}$ means joint density function of $\boldsymbol{Y}_t(k_t)$.

  \item[v.] 
      $\boldsymbol{y}_t(k_t) \perp R\,$ for all $t=1, \ldots, \tau$ if and only if $\theta_2=0$.
\item[vi.] (Inheritance Property) Consider two distribution functions
\[
H^{[1]}\left( \boldsymbol{z}_t(k_t)|r\right):=H_t^\#\left( \boldsymbol{z}_t(k_t)|r\right)
\quad\hbox{and}\quad
H^{[2]}\left( \boldsymbol{z}_t(k_t^*)|r\right):=H_t^\#\left( \boldsymbol{z}_t(k_t^*)|r\right)
\]
for $k_t\le k_t^*$. Then, we have the following inheritance property
\begin{equation}\label{eq.c2}
H^{[1]}\left( n_t, y_{t,1}, \cdots, y_{t, k_t}|r\right)=
\lim\limits_{y_{t, k_t+1}\rightarrow \infty, \cdots, y_{t, k_t^*}\rightarrow \infty}
H^{[2]}\left( n_t, y_{t,1}, \cdots, y_{t, k_t}, y_{t, k_t+1}, \cdots, y_{t, k_t^*}|r\right)
\end{equation}
for any $\boldsymbol{z}_t(k_t)=\left(n_t, y_{t,1}, \cdots, y_{t, k_t}\right)$ and $r$.

\end{enumerate}
\end{model}

Note that part v in Model \ref{model.4} is necessary for the well-definedness of the model since the model is repeatedly defined for multiple times for \eqref{eq.c1}.
One immediate result from Model \ref{model.4} is that its density function  at $\boldsymbol{Z}_t(k_t)=\boldsymbol{z}_t(k_t)$
\[
h_t^\#(\boldsymbol{z}_{t}(k_t))=\int  h_t^\#{(n_t, \boldsymbol{y}_{t}(k_t) |r)} \pi(r){\rm d}{r}
\]
is well-defined under the classical multivariate analysis with the following relation with
the corresponding joint distribution function
\[
H_t^\#(\boldsymbol{z}_{t}(k_t))= \sum\limits_{x_0=0}^{n_t} \int_{-\infty}^{y_{t, 1}}\cdots\int_{-\infty}^{y_{t, k_t}} h_t^\#(x_0, x_1, \cdots, x_{k_t}) {\rm d}x_1\cdots {\rm d}x_{k_t}.
\]
Furthermore, importantly, we observe that the density function $h_t^\#$   at $\boldsymbol{Z}_t(k_t)=\left(n_t, y_{t,1}, \cdots, y_{t, k_t} \right)$ in Model \ref{model.4}
 coincides with the density function $h_t$ in Model \ref{model.1} at $\boldsymbol{Z}_t=\left(n_t, y_{t,1}, \cdots, y_{t, n_t} \right)$ if $k_t=n_t$. Hence, as long as inheritance property in part vi of Model \ref{model.4} holds, we can see that the density function and the corresponding distribution function in Model \ref{model.1} is well-defined having Model \ref{model.4} as background model. For the simplicity of the presentation, this paper only present Model \ref{model.1} without specifying the background model in Model \ref{model.4}.

\bigskip
\bibliographystyle{apalike}
\bibliography{My_Bib}

\begin{thebibliography}{}

\bibitem[Baumgartner et~al., 2015]{baumgartner2015sharedre}
Baumgartner, C., Gruber, L.~F., and Czado, C. (2015).
\newblock Bayesian total loss estimation using shared random effects.
\newblock {\em Insurance: Mathematics and Economics}, 62:194--201.

\bibitem[Cossette et~al., 2019]{Cossette2019}
Cossette, H., Marceau, E., and Mtalai, I. (2019).
\newblock Collective risk models with dependence.
\newblock {\em Insurance: Mathematics and Economics}, 87:153--168.

\bibitem[Czado et~al., 2012]{czado2012mixed}
Czado, C., Kastenmeier, R., Brechmann, E.~C., and Min, A. (2012).
\newblock A mixed copula model for insurance claims and claim sizes.
\newblock {\em Scandinavian Actuarial Journal}, 2012(4):278--305.

\bibitem[Denuit et~al., 2007]{Denuit2}
Denuit, M., Mar{\'e}chal, X., Pitrebois, S., and Walhin, J.-F. (2007).
\newblock {\em Actuarial modelling of claim counts: Risk classification,
  credibility and bonus-malus systems}.
\newblock John Wiley \& Sons.

\bibitem[Frees et~al., 2014a]{Frees2}
Frees, E.~W., Derrig, R.~A., and Meyers, G. (2014a).
\newblock {\em Predictive Modeling Applications in Actuarial Science},
  volume~1.
\newblock Cambridge University Press.

\bibitem[Frees et~al., 2011a]{Frees2011health}
Frees, E.~W., Gao, J., and Rosenberg, M.~A. (2011a).
\newblock Predicting the frequency and amount of health care expenditures.
\newblock {\em North American Actuarial Journal}, 15(3):377--392.

\bibitem[Frees et~al., 2016]{Gee2016}
Frees, E.~W., Lee, G., and Yang, L. (2016).
\newblock Multivariate frequency-severity regression models in insurance.
\newblock {\em Risks}, 4(1):4.

\bibitem[Frees et~al., 2011b]{frees2011summarizing}
Frees, E.~W., Meyers, G., and Cummings, A.~D. (2011b).
\newblock Summarizing insurance scores using a gini index.
\newblock {\em Journal of the American Statistical Association},
  106(495):1085--1098.

\bibitem[Frees et~al., 2014b]{frees2014insurance}
Frees, E.~W., Meyers, G., and Cummings, A.~D. (2014b).
\newblock Insurance ratemaking and a gini index.
\newblock {\em Journal of Risk and Insurance}, 81(2):335--366.

\bibitem[Garrido et~al., 2016]{Garrido}
Garrido, J., Genest, C., and Schulz, J. (2016).
\newblock Generalized linear models for dependent frequency and severity of
  insurance claims.
\newblock {\em Insurance: Mathematics and Economics}, 70:205 -- 215.

\bibitem[Genest and Ne{\v{s}}lehov{\'a}, 2007]{genest2007primer}
Genest, C. and Ne{\v{s}}lehov{\'a}, J. (2007).
\newblock A primer on copulas for count data.
\newblock {\em ASTIN Bulletin: The Journal of the IAA}, 37(2):475--515.

\bibitem[Genz and Bretz, 2009]{Genz2009}
Genz, A. and Bretz, F. (2009).
\newblock {\em Computation of multivariate normal and t probabilities}, volume
  195.
\newblock Springer Science \& Business Media.

\bibitem[Hern{\'a}ndez-Bastida et~al., 2009]{hernandez2009net}
Hern{\'a}ndez-Bastida, A., Fern{\'a}ndez-S{\'a}nchez, M., and
  G{\'o}mez-D{\'e}niz, E. (2009).
\newblock The net bayes premium with dependence between the risk profiles.
\newblock {\em Insurance: Mathematics and Economics}, 45(2):247--254.

\bibitem[Jeong and Valdez, 2020]{jeong2019predictive}
Jeong, H. and Valdez, E.~A. (2020).
\newblock Predictive compound risk models with dependence.
\newblock {\em {SSRN},
  \url{https://papers.ssrn.com/sol3/papers.cfm?abstract_id=3494575}}.

\bibitem[Kadhem and Nikoloulopoulos, 2019]{kadhem2019factor}
Kadhem, S.~H. and Nikoloulopoulos, A.~K. (2019).
\newblock Factor copula models for mixed data.
\newblock {\em arXiv preprint arXiv:1907.07395}.

\bibitem[Klugman et~al., 2012]{Klugman}
Klugman, S.~A., Panjer, H.~H., and Willmot, G.~E. (2012).
\newblock {\em Loss models: from data to decisions}, volume 715.
\newblock John Wiley \& Sons.

\bibitem[Kr{\"a}mer et~al., 2013]{kramer2013}
Kr{\"a}mer, N., Brechmann, E.~C., Silvestrini, D., and Czado, C. (2013).
\newblock Total loss estimation using copula-based regression models.
\newblock {\em Insurance: Mathematics and Economics}, 53(3):829--839.

\bibitem[Krupskii and Joe, 2013]{krupskii2013factorcopula}
Krupskii, P. and Joe, H. (2013).
\newblock Factor copula models for multivariate data.
\newblock {\em Journal of Multivariate Analysis}, 120:85--101.

\bibitem[Krupskii and Joe, 2015]{krupskii2015structured}
Krupskii, P. and Joe, H. (2015).
\newblock Structured factor copula models: Theory, inference and computation.
\newblock {\em Journal of Multivariate Analysis}, 138:53--73.

\bibitem[Landsman and Valdez, 2003]{landsman2003tail}
Landsman, Z.~M. and Valdez, E.~A. (2003).
\newblock Tail conditional expectations for elliptical distributions.
\newblock {\em North American Actuarial Journal}, 7(4):55--71.

\bibitem[Lee and Shi, 2019]{Leegee}
Lee, G.~Y. and Shi, P. (2019).
\newblock A dependent frequency--severity approach to modeling longitudinal
  insurance claims.
\newblock {\em Insurance: Mathematics and Economics}, 87:115--129.

\bibitem[Lee et~al., 2020]{lee2020poisson}
Lee, W., Kim, J., and Ahn, J.~Y. (2020).
\newblock The poisson random effect model for experience ratemaking:
  Limitations and alternative solutions.
\newblock {\em Insurance: Mathematics and Economics}, 91:26--36.

\bibitem[Murray et~al., 2013]{Murray2013}
Murray, J.~S., Dunson, D.~B., Carin, L., and Lucas, J.~E. (2013).
\newblock Bayesian gaussian copula factor models for mixed data.
\newblock {\em Journal of the American Statistical Association},
  108(502):656--665.

\bibitem[Nikoloulopoulos and Joe, 2015]{nikoloulopoulos2015factor}
Nikoloulopoulos, A.~K. and Joe, H. (2015).
\newblock Factor copula models for item response data.
\newblock {\em Psychometrika}, 80(1):126--150.

\bibitem[Oh et~al., 2020a]{Oh2019copula}
Oh, R., Ahn, J.~Y., and Lee, W. (2020a).
\newblock On copula-based collective risk models: from elliptical copulas to
  vine copulas.
\newblock {\em Scandinavian Actuarial Journal}, In Press.

\bibitem[Oh et~al., 2020b]{PengAhn}
Oh, R., Shi, P., and Ahn, J.~Y. (2020b).
\newblock Bonus-malus premiums under the dependent frequency-severity modeling.
\newblock {\em Scandinavian Actuarial Journal}, 2020(3):172--195.

\bibitem[Park et~al., 2018]{park2018}
Park, S.~C., Kim, J.~H., and Ahn, J.~Y. (2018).
\newblock Does hunger for bonuses drive the dependence between claim frequency
  and severity?
\newblock {\em Insurance: Mathematics and Economics}, 83:32--46.

\bibitem[Shi et~al., 2015]{Peng}
Shi, P., Feng, X., and Ivantsova, A. (2015).
\newblock Dependent frequency--severity modeling of insurance claims.
\newblock {\em Insurance: Mathematics and Economics}, 64:417--428.

\bibitem[Shi and Yang, 2018]{shi2018pair}
Shi, P. and Yang, L. (2018).
\newblock Pair copula constructions for insurance experience rating.
\newblock {\em Journal of the American Statistical Association},
  113(521):122--133.

\bibitem[Yang et~al., 2019]{yang2019nonparametric}
Yang, L., Frees, E.~W., and Zhang, Z. (2019).
\newblock Nonparametric estimation of copula regression models with discrete
  outcomes.
\newblock {\em Journal of the American Statistical Association}, pages 1--25.

\end{thebibliography}

\end{document}